\documentclass[11pt]{article}
\usepackage{fullpage}
\usepackage[bookmarks]{hyperref}
\usepackage{amssymb}
\usepackage{amsmath}
\usepackage{amsthm}
\usepackage{graphicx,color,colordvi}
\usepackage{bbm}
\usepackage{cleveref}
\usepackage{stmaryrd}
\usepackage[utf8]{inputenc}
\usepackage{fullpage}
\usepackage[blocks]{authblk}
\usepackage{mathtools}

\usepackage{pgfplots}

\def\openone{\leavevmode\hbox{\small1\kern-3.8pt\normalsize1}}

\def\RR{\mathbb{R}}

\def\NN{\mathbb{N}}
\def\LL{\mathbb{L}}

\def\11{\mathbb{I}}
\def\LL{\mathcal{L}}

\newtheorem{theorem}{Theorem}
\newtheorem{lemma}{Lemma}
\newtheorem{proposition}{Proposition}
\newtheorem{corollary}{Corollary}

\theoremstyle{definition}

\def\reff#1{(\ref{#1})}
\def\eps{\varepsilon}

\newcommand{\supp}{\mathop{\rm supp}\nolimits}

\newcommand{\tr}{\mathop{\rm Tr}\nolimits}

\usepackage{pst-node}
\usepackage{tikz-cd}

\newcommand{\cA}{{\cal A}}
\newcommand{\cB}{{\cal B}}
\newcommand{\cD}{{\cal D}}

\newcommand{\cH}{{\cal H}}

\newcommand{\cJ}{{\cal J}}

\newcommand{\cP}{\mathcal{P}}

\def\e{\mathrm{e}}

\usepackage{graphicx}
\usepackage{setspace}
\usepackage{verbatim}
\usepackage{subfig}

\theoremstyle{definition}

\theoremstyle{remark}

\numberwithin{equation}{section}

\DeclareRobustCommand\openone{\leavevmode\hbox{\small1\normalsize\kern-.33em1}}

\newcommand{\id}{\rm{id}}
\newcommand{\be}{\begin{equation}}
	\newcommand{\ee}{\end{equation}}
\newcommand{\bea}{\begin{eqnarray}}
	\newcommand{\eea}{\end{eqnarray}}
\newcommand{\beas}{\begin{eqnarray*}}
	\newcommand{\eeas}{\end{eqnarray*}}

\title{Relating relative entropy, optimal transport \\and Fisher information: a quantum HWI inequality.}

\author[1]{Cambyse Rouz\'{e}}
\author[1,2]{Nilanjana Datta}

\affil[1]{\small Statistical Laboratory, Centre for Mathematical Sciences, University of Cambridge, Cambridge~CB30WB, UK}
\affil[2]{\small DAMTP, Centre for Mathematical Sciences, University of Cambridge, Cambridge~CB30WA, UK}
\begin{document}

 	\bibliographystyle{abbrv}

 	\maketitle	
\begin{abstract}
Quantum Markov semigroups characterize the time evolution of an important class of open quantum systems.
Studying convergence properties of such a semigroup, and determining concentration properties of its invariant
state, have been the focus of much research. Quantum versions of functional inequalities (like the modified logarithmic Sobolev and Poincar\'e inequalities) and the so-called transportation cost inequalities, have proved to be essential for this purpose.
Classical functional and transportation cost inequalities are seen to arise from a single geometric inequality, 
called the Ricci lower bound, via an inequality which interpolates between them. The latter is called the HWI-inequality,
where the letters I, W and H are, respectively, acronyms for the Fisher information (arising in the modified logarithmic Sobolev
inequality), the so-called Wasserstein distance (arising in the transportation cost inequality) and the relative entropy (or Boltzmann
H function) arising in both. Hence, classically, all the above inequalities and the implications between them form a remarkable picture which
relates elements from diverse mathematical fields, such as Riemannian geometry, information theory, optimal transport theory, 
Markov processes, concentration of measure, and convexity theory. Here we consider a quantum version of the Ricci lower bound
introduced by Carlen and Maas, and prove that it implies a quantum HWI inequality from which the quantum functional and transportation cost inequalities follow.
Our results hence establish that the unifying picture of the classical setting carries over to the quantum one.
\end{abstract}

\section{Introduction}

Realistic physical systems which are relevant for quantum information processing are inherently open. They undergo unwanted but unavoidable interactions with the surrounding environment, and are hence subject to noise and decoherence.
Under the Markovian approximation, which is valid when the system is only weakly coupled to its environment, the resulting dissipative dynamics of the system is described by a quantum Markov semigroup (QMS), whose generator we denote by ${\cal{L}}$. The analysis of quantum Markov semigroups is hence a key component of the theory of open quantum systems and quantum information. An important problem in the study of a QMS is the analysis of its convergence properties, in particular, its {\em{mixing time}}, which is the time taken by any state evolving under the action of the QMS to come close to its invariant state\footnote{Here we assume that the QMS is primitive, i.e.~it has a unique invariant state.}. \\\\
{\bf{Functional and transportation cost inequalities:}} Classically, given a measure $\mu$, functional inequalities, e.g.~the \textit{Poincar\'{e} inequality} (usually denoted as PI($\lambda$)) \cite{TKRWV10} and the \textit{(modified) logarithmic Sobolev inequality} (or log-Sobolev in short), denoted as MLSI($\alpha$) \cite{[OZ99]}, constitute a powerful tool for deriving mixing times of a Markov semigroup with invariant measure $\mu$, and determining concentration properties of $\mu$. They are also related to the so-called \textit{transportation-cost} inequalities denoted by TC$_1(c_1)$ and TC$_2(c_2)$. Here $\alpha$, $c_1$ and $c_2$ denote constants appearing in the respective inequalities. Consider a compact manifold $\cal{M}$, and let ${\cal{P}}(\cal{M})$ be the set of probability measures on $\cal{M}$. Given a measure $\mu \in {\cal{P}}(\cal{M})$, the inequality TC$_1(c_1)$ (resp.~TC$_2(c_2)$) provides an upper bound on the so-called Wasserstein distance $W_1$ (resp.~$W_2$), between any probability measure $\nu \in {\cal{P}}(\cal{M})$ and the given measure $\mu$, in terms of the square root of the relative entropy of $\nu$ with respect to $\mu$. Since $\mu$ is fixed, this relative entropy is simply a functional of $\nu$ and, due to its close links with the Boltzmann H-functional, is often denoted by the letter H in the literature. The notion of Wasserstein distances first appeared in the theory of optimal transport, which was initiated by Monge~\cite{[Monge]} and later analyzed by Kantorovich~\cite{[Kantorovich]}. In its original formulation by Monge, the problem of optimal transport concerns finding the optimal way, in the sense of minimal transportation cost, of moving a sand pile between two locations (see also \cite{[V08]}). 
In 1986, Marton \cite{[M86]} showed that transportation-cost inequalities are also useful for deriving {\em{concentration of measure}} properties of the given measure $\mu${\footnote{Given a metric space $({\cal{X}},d)$, a probability measure $\mu$ is said to satisfy Gaussian (resp.~exponential) concentration on it if there exist positive constants $a,b$ such that for any $A \subseteq {\cal{X}}$, and $r>0$, $$ \mu(A) \geq 1/2 \,\,\implies \,\, \mu(A_r) \geq 1- a e^{-bf(r)}.$$
where $A_r :=\{x \in {\cal{X}} \,: \, d(x, A) < r\}$ and $f(r) =$ $r^2$ (resp. $r$).}}.\\\\
The classical inequalities discussed above can be shown to be obtainable from a single geometric inequality, involving a quantity called the {\em{Ricci curvature}} of the manifold ${\cal{M}}$.
In fact, there is an inherent relation between the geometry of the manifold and a diffusion process (whose associated Markov semigroup has generator $\cal{L}$, say) defined on it: the diffusion process can be used to explore the geometry of ${\cal{M}}$, and conversely, the latter determines the mixing time of the diffusion process. 
Finding a quantum analogue of this appealing geometric inequality is hence a problem of fundamental interest, and is considered in this paper. Before we present our results on this problem, we first need to explain the statement of the Ricci lower bound in the classical setting. In fact, it is instructive to start from the very definition of curvature which generalizes to the Ricci curvature for the case of a Riemannian manifold.\\\\
{\bf{Ricci curvature and Ricci lower bound (classical setting):}} Given a surface $\mathcal{S}$ embedded in the Euclidean space $\mathbb{R}^3$, the \textit{Gauss curvature} $\kappa$ of $\mathcal{S}$ is a measure of its local bendedness. More precisely, given a point $x\in\mathcal{S}$, and any two orthogonal unit tangent vectors $u,v$ at $x$, the distance between two geodesics $\gamma_u$ and $\gamma_v$, starting at $x$, with respective directions $u$ and $v$, obeys the following Taylor expansions:
\begin{align}\label{sect}
	d_g(\gamma_\mathbf{u}(t),\gamma_\mathbf{v}(t))=\sqrt{2}t\left(  1-\frac{\kappa(x)}{12}t^2+\mathcal{O}_{t\to 0}(t^3)\right), \quad t \geq 0,
\end{align}	
where $d_g$ is the geodesic distance defined with respect to the metric $g$ induced on $\mathcal{S}$ by the Euclidean metric. In the case when $\kappa=0$ uniformly on the surface, the latter is flat and we recover the Pythagoras theorem from \Cref{sect}.
\begin{figure}[!htbp]
	\centering
	\includegraphics[width=0.2\textwidth]{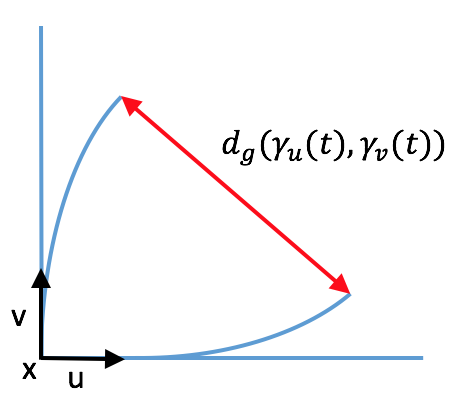}
	\caption{The Gauss curvature.}
	\label{fig3}
\end{figure}
 More generally, let $x$ be a point in a $d$-dimensional compact Riemannian manifold $\mathcal{M}$, let $u$ belong to the tangent space $T_x\mathcal{M}$ at the point $x$ of $\mathcal{M}$, and complete the vector $u$ into an orthonormal basis $(u,v_2,...,v_d)$ of $T_x\mathcal{M}$. Then, the \textit{Ricci curvature} of ${\mathcal{M}}$, evaluated at $u$ is the averaged Gauss curvature over orthogonal surfaces defined by all the geodesics starting at $x$ with direction given by the unit vectors belonging to the vector subspace spanned by $u$ and any other vector $v_i$, $i=2,...,d$. The expression for the Ricci curvature \cite{[V08]} is given in terms of the Laplace-Beltrami operator (denoted simply as $\Delta$) and hence the curvature is usually denoted as $\operatorname{Ric}(\Delta)$. Since $\Delta$ is the generator of the heat semigroup, the curvature provides a bridge between the geometry of the manifold and the evolution on it induced by the heat diffusion. There is an important inequality, known as the \textit{Ricci lower bound}, which is denoted by $\operatorname{Ric}(\Delta)\ge \kappa$ \cite{[BE85]}, and is the property that the Ricci curvature is uniformly bounded below by a real parameter $\kappa\ge 0$. Intuitively, the inequality is related to concentration of the uniform measure on $\cal{M}$, hich is known to be the unique invariant measure of heat diffusion (whose generator is $\Delta$). For example, in the case of the sphere, which has constant Ricci curvature given in terms of its radius, the Haar measure can be shown to concentrate around any great circle. One can relax the condition of uniformity of the measure in order to allow for the study of concentration of measure phenomena for different measures $\mu$, invariant for other diffusions processes on $\mathcal{M}$. In this more general framework, the Ricci lower bound is denoted by $\operatorname{Ric}(\mathcal{L})\ge \kappa$, where $\mathcal{L}$ denotes the generator of the diffusion semigroup associated to $\mu$.\\\\
More recently, 
Sturm \cite{sturm2006geometry,Sturm2006} and Lott-Villani \cite{lott2009ricci} showed that $\operatorname{Ric}(\mathcal{L})\ge \kappa$, can be viewed as a (refined) convexity property (called the $\kappa$-{\em{displacement convexity}}) of H along geodesics on the Riemannian manifold obtained by endowing the set of probability measures $\mathcal{P}(\mathcal{M})$ on $\mathcal{M}$ with the Wasserstein distance $W_2$ \cite{von2005transport}. This discovery led to a more robust notion of a Ricci lower bound which does not explicitly depend on the expression of the Ricci curvature, and hence can be extended to more general metric spaces. Starting from this convexity property, one can then construct a diffusion semigroup for which H decreases the most along the direction of evolution induced by the semigroup. In this case the path on the Riemannian manifold $({\cal{P}}({\cal{M}}), W_2)$, which corresponds to the actual evolution under the diffusion, is said to be {\em{gradient flow}} for H. It is a striking fact that this diffusion coincides with the one whose generator appears in the Bakry-\'{E}mery condition (see \cite{jordan1998variational,erbar2010}).\\\\
In \cite{[OV00]}, the authors introduced the so-called HWI$(\kappa)$-{\em{interpolation inequality}}, using which they reproved the so-called Bakry-\'{E}mery theorem, which states that for $\kappa>0$, $\operatorname{Ric}(\mathcal{L})\ge \kappa$ implies MLSI($\alpha$) (for diffusions on $\mathbb{R}^n$ with associated with generator $\cal{L}$). The letters W, I and H are, respectively, acronyms for the Wasserstein distance $W_2$ (appearing in TC$_2(c_2)$), the Fisher information (which arises in MLSI($\alpha$)) and the relative entropy (also called the Boltzmann H-functional, as mentioned above) which appears in both these inequalities. They also showed that MLSI($\alpha$) implies TC$_2(c_2)$. The term interpolation here comes from the fact that in the case $\kappa=0$ and $c>0$, TC$_2$($c$) together with HWI$(0)$ gives back MLSI($\alpha$).\\\\
In \cite{maas2011gradient,[EM12],erbar2016poincar}, a modified version of the Ricci lower bound was defined for Markov processes on finite sets, which led to the unification of the previously discussed functional and concentration inequalities in this discrete framework. In particular, it was proved in \cite{erbar2016poincar} that one can recover the Poincar\'{e} and modified log-Sobolev inequalities from the Ricci lower bound, provided the diameter of ${\cal{P}}({\cal{M}})$, with respect to the Wasserstein distance, $W_2$, is bounded.\\\\
{\bf{Ricci lower bound (quantum setting):}} In the case of a quantum system with a finite-dimensional Hilbert space ${\cal{H}}$, the set $\mathcal{P}(\mathcal{M})$ is replaced by the set ${\cal{D}}({\cal{H}})$ of quantum states (i.e.~density matrices) on ${\cal{H}}$. Then, in analogy with the classical case, starting with a primitive QMS with generator ${\cal{L}}$, Carlen and Maas \cite{[CM14],Carlen20171810} defined 
a quantum Wasserstein distance $W_{2,\mathcal{L}}$ which renders ${\cal{D}}({\cal{H}})$ with a Riemannian structure, and for which the master equation associated to the QMS is gradient flow for the quantum relative entropy.\\\\
In \cite{Carlen20171810}, the authors proved that a quantum MLSI$(\alpha)$, first introduced in \cite{[KT13]}, holds provided the quantum relative entropy (between a state on a geodesic on this manifold and the invariant state of the QMS) satisfies a quantum analogue of the $\kappa$-displacement convexity property along geodesics, for $\alpha=\kappa>0$. This is denoted below by Ric$(\mathcal{L})\ge \kappa$ in analogy with the classical case, with $\cal{L}$ being the generator of the QMS.\\\\
The quantum versions of the Ricci lower bound, the HWI inequality, and the functional and transportation cost inequalities,
all fit into a unifying picture which is analogous to the classical setting. It is given in the following figures.
\begin{figure}[!ht]
	\[
	\begin{tikzcd}
	\text{Ric\,($\mathcal{L}$)$\ge \kappa$}\arrow[r, Rightarrow,swap, ]&\text{HWI\,($\kappa$)}\arrow[rrr,Rightarrow,bend left,swap,"W_{2,\mathcal{L}}\le D\text{, }\lambda\propto D^{-2}",swap]\arrow[r,Rightarrow,bend left,swap,"\alpha=\kappa>0"] \arrow[r,Rightarrow,bend right,swap,"W_{2,\,\mathcal{L}}<D\text{, }\alpha\propto D^{-2}\text{, }\atop{(\Lambda_t)_{t\ge 0}\text{ unital}}"] &\text{MLSI}(\alpha)\arrow[r,Rightarrow,"c_2=\alpha^{-1}"]&\text{TC}_2(c_2)\arrow[d,Rightarrow,"c_1=d\,c_2",swap]\arrow[r,Rightarrow,"\lambda\propto c_2",swap]&\text{PI}(\lambda)\arrow[r,Rightarrow]&\text{Exp.}\\
	&&&\text{TC}_1(c_1)\arrow[r,Rightarrow]&\text{Gauss.}
	\end{tikzcd}
	\]
	\[\text{TC}_2(c_2)+\text{HWI}(\kappa)\Rightarrow \text{MLSI}(\alpha),~~\alpha=\max\left[ \frac{1}{4c_2}(1+c_2\kappa)^2,\kappa   \right].\]
	\caption{Chain of quantum functional- and Talagrand inequalities and related concentrations for a primitive semigroup $(\Lambda_t)_{t\ge 0}$ with generator $\LL$ defined on a Hilbert space of dimension $d$. The implication MLSI$(\alpha)$ $\Rightarrow$ PI$(\lambda)$ was proved in~\cite{[KT13]}. Here, ``Exp.'' refers to the notion of exponential concentration, whereas ``Gauss.'' refers to the stronger notion of Gaussian concentration. The implications MLSI($\alpha$)$\Rightarrow$TC$_2$($c_2$)$\Rightarrow$PI($\lambda$)$\Rightarrow$Exp., as well as TC$_2$($c_2$)$\Rightarrow$TC$_1$($c_1$)$\Rightarrow$Gauss.~were proved in \cite{rouze2017concentration}.}
	\label{fig2}
\end{figure}
\paragraph{Our contribution:}

In this paper, we analyse the quantum version of the Ricci lower bound introduced by Carlen and Maas \cite{Carlen20171810}, and derive various implications of it in \Cref{theofund}. Moreover, we show that $\operatorname{Ric}(\LL)\ge \kappa$ implies a quantum version of the celebrated $\operatorname{HWI}(\kappa)$ inequality which interpolates between the modified logarithmic Sobolev inequality and the transportation cost inequality (\Cref{theHWI}). We show that, in the case of $\kappa>0$, $\operatorname{HWI}(\kappa)\Rightarrow \operatorname{MLSI}(\kappa)$ (\Cref{theo5}), recovering the result of \cite{Carlen20171810}. 
  On the other hand, in \Cref{cor1}, we establish that in the case when $\kappa\in\RR$, Ric($\LL\ge \kappa$) together with TC$_2(c_2)$ imply MLSI$(\alpha)$. Moreover, in the case when $\kappa= 0$, we show that, under the assumption of boundedness of the diameter $D$ of the set of states with respect to the quantum Wasserstein distance $W_{2,\LL}$,~ $\operatorname{Ric}(\LL)\ge 0$ implies PI$(c_1D^{-2})$ for some universal constant $c_1$ (\Cref{poinca}). Moreover, in the case of a unital QMS (i.e.~one which has the completely mixed state as its unique invariant state), we show that it also implies MLSI$(c_2D^{-2})$ for some universal positive constant $c_2$ (\Cref{mlsi1}). We hence extend the results of \cite{erbar2016poincar} to the quantum regime.

\paragraph{Layout of the paper:}
In \Cref{sec2}, we introduce the necessary notations and definitions, including quantum Markov semigroups, the quantum Wasserstein distance and quantum functional inequalities. The quantum version of $\kappa$-displacement convexity is studied in \Cref{sec3}. In \Cref{333}, we prove the quantum HWI$(\kappa)$ inequality, show that it implies $\operatorname{MLSI}(\kappa)$ in the case when $\kappa>0$ and derive interpolation results between $\operatorname{MLSI}(\kappa)$ and $\operatorname{TC}_2(c_2)$ from it. In \Cref{poinc}, we show that in the case in which $\kappa=0$, $\operatorname{PI}(\lambda)$ holds with a constant $\lambda$ proportional to $D^{-2}$, where $D$ stands for the diameter of the set of states. In \Cref{mls}, we show that under the further assumption of the QMS being unital, $\operatorname{MLSI}(\alpha_1)$ holds with a constant $\alpha_1$ also proportional to $D^{-2}$.

\section{Notations and preliminaries}\label{sec2}

\subsection{Operators, states and entropic quantities}
In this paper, we denote by $(\cH,\langle .|.\rangle)$ a finite-dimensional Hilbert space of dimension $d$ with associated inner product $\langle.|.\rangle$, by $\cB(\cH)$ the algebra of linear operators acting on $\cH$, and by $\cB_{sa}(\cH) \subset \cB(\cH)$ the subspace of self-adjoint operators. Moreover, the Hilbert Schmidt inner product $\langle .,.\rangle$, where $\langle A,B\rangle=\tr(A^*B)$ $\forall A,B\in\cB(\cH)$, provides $\cB(\cH)$ with a Hilbert space structure. Let $\cP(\cH)$ be the cone of positive semi-definite operators on $\cH$ and $\cP_{+}(\cH) \subset  \cP(\cH)$ the set of (strictly) positive operators. Further, let $\cD(\cH):=\lbrace\rho\in\cP(\cH)\mid \tr\rho=1\rbrace$ denote the set of density operators (or states) on $\cH$, and $\cD_+(\cH):=\cD(\cH)\cap \cP_+(\cH)$ denote the subset of faithful states. We denote the support of an operator $A$ by ${\mathrm{supp}}(A)$. Let $\mathbb{I}\in\cP(\cH)$ be the identity operator on $\cH$, and $\id:\cB(\cH)\mapsto \cB(\cH)$ the identity map on operators on~$\cH$. For $p,q\ge 1$, the $p$-Schatten norm of an operator $A\in\cB(\cH)$ is denoted by $\|A\|_p:=(\tr|A|^p)^{1/p}$, and the $p\to q$-norm of a superoperator $\Lambda:\cB(\cH)\to\cB(\cH)$ by $\|\Lambda\|_{p\to q}$. Such a linear map is said to be unital if $\Lambda(\mathbb{I})=\mathbb{I}$. Given two states $\rho,\sigma\in\cD(\cH)$, the quantum relative entropy between $\rho$ and $\sigma$ is defined as:
\begin{align*}
	{D}(\rho\|\sigma):= \left\{\begin{aligned}
		&	\tr(\rho(\log\rho-\log\sigma))~~~~~\text{if} \supp(\rho)\subseteq \supp(\sigma),\\
		&	0~~~~~~~~~~~~~~~~~~~~~~~~~~~~~\text{else.}
	\end{aligned}
	\right.
\end{align*}
\subsection{Quantum Markov semigroups and the detailed balance condition}
 In the Heisenberg picture, a quantum Markov semigroup (QMS) on a finite dimensional Hilbert space $\mathcal{H}$ it is given by a one-parameter family $\left(\Lambda_t\right)_{t \geq 0}$ of linear, completely positive, unital maps on $\cB(\mathcal H)$ satisfying the following properties
\begin{itemize}
	\item $\Lambda_0 = {\rm{id}}$;
	\item $\Lambda_t \Lambda_s = \Lambda_{t+s}$ $\,-\,$ semigroup property;
	\item $\forall X\in\cB(\cH),~\underset{t\to 0}{\lim} || \Lambda_t(X) - X||_\infty = 0$ $\,-\,$ strong continuity.
\end{itemize}
The parameter $t$ plays the role of time. For each quantum Markov semigroup there exists an operator $\LL$ called the \textit{generator}, or \textit{Lindbladian}, of the semigroup, such that
\begin{align}\label{llambda}
	\frac{d}{dt}\Lambda_t=\Lambda_t\circ \LL=\LL\circ\Lambda_t.
\end{align}
In the Schr\"{o}dinger picture, the dual of $\Lambda_t$ is written $\Lambda_{*t}$, for any $t\ge 0$. The QMS is said to be \textit{primitive} if there exists a unique invariant state $\sigma$ i.e., such that $\Lambda_{*t}(\sigma)=\sigma$. Such a QMS is said to satisfy the \textit{detailed balance condition} if the following holds:
\begin{align}\label{DBC}
	\tr(\sigma \LL(X)^*Y)=\tr(	\sigma  X^*\LL(Y)),~~~~~~~~~X,Y\in\cB(\cH).
\end{align}
In the context of quantum logarithmic Sobolev inequalities, the \textit{quantum Fisher information} of $\rho$ with respect to the sate $\sigma$, first defined in \cite{[S78]}, is particularly useful:\footnote{This quantity is also referred to as \textit{entropy production}.}
\begin{align*}
	\operatorname{I}_\sigma(\rho):=\left\{\begin{aligned}
		&-\tr(\LL_*(\rho)(\log \rho-\log\sigma) ), ~~~~~~~~~~ \rho\in\cD_+(\cH)\\
		&+\infty,~~~~~~~~~~~~~~~~~~~~~~~~~~~~~~~~~~~~~~\text{otherwise}.
	\end{aligned}\right.
\end{align*}
 The following theorem provides a structure for the generators of primitive QMS satisfying the detailed balance condition:
				\begin{theorem}[\cite{A76,Carlen20171810}]Let $\sigma\in\cD_+(\cH)$, and let $(\Lambda_t)_{t\ge0}$ be a quantum Markov semigroup on $\cB(\cH)$. Suppose that the generator $\LL$ of  $(\Lambda_t)_{t\ge0}$ satisfies the detailed balance condition. Then there exists an index set $\mathcal{J}$ of cardinality $|\mathcal{J}|\le d^2-1$ such that $\LL$ takes the form 
					\begin{align}\label{LLDBC}
						\LL(f)&=\sum_{j\in \mathcal{J}}c_j\left( \e^{-\omega_j/2}\tilde{L}_j^*[f,\tilde{L}_j]+\e^{\omega_j/2}[\tilde{L}_j,f]\tilde{L}_j^*\right)
					\end{align}
					where $\omega_j\in\RR$ and $c_j>0$ for all $j\in\mathcal{J}$, and $\{\tilde{L}_j\}_{j\in\mathcal{J}}$ is a set of operators in $\cB(\cH)$ with the properties:
					\begin{itemize}
						\item[1] $\frac{1}{\dim(\cH)}\tr(\tilde{L}_j^*\tilde{L}_k)=\delta_{k,j}$ for all $j,k\in\mathcal{J}$
						\item[2] $\tr(\tilde{L}_j)=0$ for all $j\in\mathcal{J}$
						\item[3] $\{\tilde{L}_j\}_{j\in\mathcal{J}}=\{\tilde{L}_j^*\}_{j\in\mathcal{J}}$
						\item[4] $\{\tilde{L}_j\}_{j\in\mathcal{J}}$ consists of eigenvectors of the modular operator $\Delta_\rho:f\mapsto \rho f \rho^{-1}$ with
						\begin{align*}
							\Delta_\sigma(\tilde{L}_j)=\e^{-\omega_j}\tilde{L}_j.
						\end{align*}
					\end{itemize}
					Finally for each $j\in\mathcal{J}$
					\begin{align}\label{symcond}
						c_j=c_{j'}, \text{ and }~ \omega_j=-\omega_{j'}~\text{ when }\tilde{L}_j^*=\tilde{L}_{j'}.
					\end{align}
					Conversely, given any faithful state $\sigma$, any set $\{\tilde{L}_j\}_{j\in\mathcal{J}}$ satisfying the above four conditions for some $\{\omega_j\}_{j\in \mathcal{J}}\subset \RR$ and any set $\{c_j\}_{j\in\mathcal{J}}$ of positive numbers satisfying the symmetry condition \reff{symcond}, the operator $\LL$ given by \Cref{LLDBC} is the generator of a quantum Markov semigroup $(\Lambda_t)_{t\ge 0}$ which satisfies the detailed balance condition. 
				\end{theorem}	
			
			\subsection{The Wasserstein distance $W_{2,\LL}$}
			In this section, we recall the construction of the Wasserstein metric $W_{2,\LL}$ first defined in \cite{Carlen20171810}. Assume given a generator $\LL$ of a primitive QMS, with invariant state $\sigma$, of the form of \reff{LLDBC}. Given an operator $U\in \cB(\cH)$, its noncommutative \textit{gradient} is defined as:
\begin{align*}
	\nabla U:=(\partial_1U,...,\partial_{\mathcal{J}}U), ~~~ U\in\cB(\cH),
\end{align*}	
where $\partial_j X=[\tilde{L}_j,X]$ for all $j\in\mathcal{J}$. Similarly, given a vector $\mathbf{A}\equiv (A_1,...,A_{|\cJ|})\in \bigoplus_{j\in\cJ}\cB(\cH)$, the \textit{divergence} of $\mathbf{A}$ is defined as
\begin{align*}
	\operatorname{div}(\mathbf{A}):=\sum_{j\in\mathcal{J}} c_j[A_j,\tilde{L}_j^*]\equiv- \sum_{j\in\cJ}c_j \partial^*_j A_j,
\end{align*}
where $\partial_j^*X:=[\tilde{L}_j^*,X]$. For $\vec{\omega}:=(\omega_1,...,\omega_{|\mathcal{J}|})$, define the linear operator $[\rho]_{\vec{\omega}}$ on $\bigoplus_{j\in \mathcal{J}}\cB(\cH)$ through
\begin{align*}
	[\rho]_{\vec{\omega}}\mathbf{A}:=([\rho]_{\omega_1}A_1,...,[\rho]_{\omega_{|\mathcal{J}|}}A_{|\mathcal{J}|}),~~~~~~~~\mathbf{A}\equiv(A_1,...,A_{|\mathcal{J}|}),
\end{align*}
where for any $\omega\in \RR$,
\begin{align}\label{fomega}
	[\rho]_\omega=R_\rho\circ f_\omega(\Delta_{\rho}),~~~~f_\omega(t):=\e^{\omega/2}\frac{t-\e^{-\omega}}{\log t+\omega},~~~~t\in\RR,
\end{align}
where $R_\rho:\cB(\cH)\to \cB(\cH)$ denotes the operator of left multiplication by $\rho$. Intuitively, $	[\rho]_\omega$ can be understood as a noncommutative way of multiplying by $\rho$:
\begin{lemma}[see Lemma 5.8 of \cite{Carlen20171810}]\label{lemma}
	For any $\omega\in\RR$, and $\rho\in \cD_+(\cH)$, 
	\begin{align*}
		[\rho]_{\omega}(A)=\int_0^1 \e^{\omega(1/2-s)}\rho^s A\rho^{1-s}ds.
	\end{align*}
\end{lemma}
\noindent
Let $(\gamma(s))_{s\in (-\eps,\eps)}$ be a differential path in $\cD_+(\cH)$ for some $\eps>0$, and denote $\rho:=\gamma(0)$. Then $\tr(\dot{\gamma}(0))=\left.\frac{d}{ds}\right|_{s=0}\tr(\gamma(s))=0$. Carlen and Maas proved that there is a unique vector field $\mathbf{V}\in\bigoplus_{j\in\mathcal{J}} \cB(\cH) $ of the form $\mathbf{V}=\nabla U$, where $U\in\cB(\cH)$ is traceless and self-adjoint, for which the following non-commutative continuity equation holds:
\begin{align}\dot{\gamma}(0)=-\operatorname{div}([\rho]_{\vec{\omega}}\nabla U).\label{continuity2}
\end{align}
Define the inner product $\langle .,. \rangle_{\LL,\rho}$ on $\bigoplus_{j\in\mathcal{J}}\cB(\cH)$ through:
\begin{align}\label{innerj}
	\langle \mathbf{W},\mathbf{V}\rangle_{\LL,\rho}:=\sum_{j\in\mathcal{J}} c_j \langle W_j,[\rho]_{\omega_j}{V}_j\rangle ,
\end{align}
where $\langle A,B\rangle :=\tr(A^* B)$ denotes the usual Hilbert Schmidt inner product on $\cB(\cH)$. 
Hence, looking upon $\cD_+(\cH)$ as a manifold, for each $\rho\in \cD_+(\cH)$, we can identify the tangent space $T_\rho$ at $\rho$ with the set of gradient vector fields $\{\nabla U:~U\in\cB(\cH),~ U=U^*\}$ through the correspondence provided by the continuity equation (\ref{continuity2}). Defining the metric $g_\LL$ through the relation
\begin{align}\label{gammav}
	\|\dot{\gamma}(0)\|_{g_{\LL,\rho}}^2:= \|\mathbf{V}(s)\|^2_{\LL,\rho},
\end{align}
this endows the manifold $\cD_+(\cH)$ with a smooth Riemmanian structure. In this framework, Carlen and Maas then defined the \textit{modified non-commutative Wasserstein distance} $W_{2,\LL}$ to be the energy associated to the metric $g_{\LL}$, i.e.: 
\begin{align}\label{W2}
	W_{2,\LL}(\rho,\sigma):=\inf_{\gamma}\left\{ \left(\int_0^1 
	\|\mathbf{V}(s)\|^2_{\LL,\gamma(s)}
	% 		\left\|\frac{d{\gamma}(s)}{ds}\right\|_{g_{\LL,\gamma(s)}}^2
	ds\right)^{1/2}~:~\gamma(0)=\rho,~~\gamma(1)=\sigma\right\},
\end{align}
where the infimum is taken over smooth paths $\gamma:[0,1]\to \cD_+(\cH)$, and $\mathbf{V}:[0,1]\to \bigoplus_{i\in\cJ}\cB(\cH)$ is related to $\gamma$ through the continuity equation \reff{continuity2}. The paths achieving the infimum, if they exist, are the minimizing geodesics with respect to the metric $g_\LL$. The following lemma, proved in \cite{rouze2017concentration}, follows from a standard argument:
	\begin{lemma}\label{characwass}
	With the above notations, the Wasserstein distance between two faithful states $\rho,\sigma$ is equal to the minimal length over the smooth paths joining $\rho$ and $\sigma$:
	\begin{align}\label{length}
		W_{2,\LL}(\rho,\sigma)=\inf_{\gamma(s)\text{ const. speed}}\left\{\int_0^1 \|\dot{\gamma}(s)\|_{g_{\LL,\gamma(s)}}ds:~\gamma(0)=\rho,~\gamma(1)=\sigma\right\},
	\end{align}
	where the infimum is taken over curves $\gamma$ of constant speed, i.e. such that $s\mapsto \|\dot{\gamma}(s)\|_{g_{\LL,\gamma(s)}}$ is constant on $[0,1]$. 
\end{lemma}
\noindent
This definition for the quantum Wasserstein distance, $W_{2,\LL}$, is natural in the sense that the master equation
\begin{align*}
	\dot{\rho}_t=\LL_*\rho_t,
\end{align*}
is \textit{gradient flow} for $D(.\|\sigma)$, where $\sigma$ is the invariant state associated to $\LL$. This means that $\LL_* \rho=-\operatorname{grad}_\LL D(\rho\|\sigma)$, where the gradient $\operatorname{grad}_{\LL}$ of a differentiable functional $\mathcal{F}:\cD_+(\cH)\to \RR$ is defined as the unique element in the tangent space at $\rho$ so that 
\begin{align}\label{gradientflow}
	\left.	\frac{d}{dt}\mathcal{F}(\gamma(t))\right|_{t=0}=g_{\LL,\rho}(\dot{\gamma},\operatorname{grad}_{g_{\LL,\rho}}\mathcal{F}(\rho))
\end{align}
for all smooth paths $\gamma(t)$ defined on $(-\eps,\eps)$ for some $\eps>0$ with $\gamma(0)=\rho$. In particular, for $\gamma(t)=\rho_t\equiv \Lambda_{*t}(\rho)$,
\begin{align}\label{gradflow}
	\left.\frac{d}{dt}D(\rho_t\|\sigma)\right|_{t=0}=-g_{\LL,\rho}(\LL_*(\rho),\LL_*(\rho))=-\|\LL_*(\rho)\|_{g_{\LL,\rho}}^2.
\end{align}	 
The following lemma is going to play a crucial role in the rest of this paper:
\begin{lemma}\label{lemmametric}
For any $\rho\in\cD_+(\cH)$, the map $D_{\vec{\omega}}(\rho): U\mapsto -\operatorname{div}([\rho]_{\vec{\omega}} \nabla U)  $ is invertible and positive on the space of self-adjoint, traceless operators. Moreover, if $\rho\ge \eps\mathbb{I}$ for some $\eps>0$, then:
		\begin{align*}
	 D_{\vec{\omega}}(\rho)^{-1}  \le K_\LL~  \eps^{-1} ~\id,
	\end{align*}
	where $K_\LL:=\sup_{j\in\cJ}  \frac{\omega_j}{\e^{\omega_j/2}-\e^{-\omega_j/2}}\|(-\operatorname{div}\circ \nabla (.))^{-1}\|_{\infty\to\infty}>0  $.
\end{lemma}

\begin{proof}
Let $\mathcal{W}$ be the space of self-adjoint, traceless operators on $\cH$. From Theorem 7.3 of \cite{Carlen20171810}, for any $C^1$ path $(\gamma(t))_{t\in(-\eps,\eps)}$, with $\gamma(0)=\rho$, there exists a unique vector field of the form $\nabla U$ for which the continuity equation $\dot{\gamma}(0)= -\operatorname{div} ([\rho]_{\vec{\omega}}(\nabla U))$ holds. Moreover, by ergodicity of $(\Lambda_t)_{t\ge 0}$, $\ker (\nabla)$ consists of multiples of the identity. Therefore, there exists a unique $U\in \mathcal{W}$ such that $\dot{\gamma}(0)= -\operatorname{div} ([\rho]_{\vec{\omega}}(\nabla U))$. Now, for any $U,V\in \mathcal{W}$:
	\begin{align*}
		\langle U,D_{\vec{\omega}}(\rho)[V]\rangle&=\sum_{j\in\cJ} c_j \langle \partial_jU,[\rho]_{\omega_j} \partial_j V\rangle\\
		&= \sum_{j\in\cJ} c_j\langle [\rho]_{\omega_j}\partial_j U, \partial_j V\rangle\\
		&=\sum_{j\in\cJ} c_j \langle \partial_j^*([\rho]_{\omega_j} \partial_j U), V\rangle\\
		&=-\langle \operatorname{div}([\rho]_{\vec{\omega}}\nabla U),V\rangle\\
		&= \langle D_{\vec{\omega}}(\rho)[U],V\rangle,
	\end{align*}
	which means that $D_{\vec{\omega}}(\rho)$ is indeed self-adjoint. By the same argument, we can show that the superoperators $\nabla: \cB(\cH)\to \sum_{j\in\cJ}\cB(\cH) $ and $-\operatorname{div}:\oplus_{j\in\cJ}\cB(\cH)\to \cB(\cH) $ are adjoint to each other, where $\cB(\cH)$ and $\oplus_{j\in\cJ}\cB(\cH)$ are provided with the inner products $\langle .,. \rangle$ and $\sum_{j\in \cJ} c_j\, \langle .,.\rangle$, respectively. Indeed for any $U\in \cB(\cH)$ and $\mathbf{V}=(V_1,...,V_{|\cJ|})\in\oplus_{j\in\cJ} \cB(\cH)$:
\begin{align*}
	\sum_{j\in\cJ} c_j  \langle V_j	,\partial_j U\rangle= \sum_{j\in\cJ} c_j \tr ((V_j)^* [\tilde{L}_j,U]  )= \sum_{j\in\cJ} c_j \tr( [V_j^* ,\tilde{L}_j] U  )&= -  \sum_{j\in\cJ} c_j \tr( [V_j ,\tilde{L}_j^*]^* U  )\\
	&=\langle -\operatorname{div}(\mathbf{V}),U\rangle.
	\end{align*}
Assume now that $\rho\ge \eps\mathbb{I}$ for some $\eps>0$, so that for any $j\in \cJ$:
\begin{align*}
	[\rho]_{\omega_j}=\int_0^1 \e^{\omega_j(1/2-\alpha)}L_\rho^\alpha R_\rho^{1-\alpha}d\alpha \ge \eps \frac{\e^{\omega_j/2}-\e^{-\omega_j/2}}{\omega_j}\id>0.
\end{align*}	
	Hence, $	-\operatorname{div} \circ[\rho]_{\vec{\omega}}\circ \nabla( .)$ is positive and
	\begin{align*}
	(	-\operatorname{div} \circ[\rho]_{\vec{\omega}}\circ \nabla( .) )^{-1}\le \eps^{-1}\sup_{j\in\cJ}  \frac{\omega_j}{\e^{\omega_j/2}-\e^{-\omega_j/2}}   (-\operatorname{div}\circ \nabla (.))^{-1},
		\end{align*}
	and the result follows.
	\qed
\end{proof}	
\noindent
The above lemma allows us to extend the definition of the Wasserstein distance to non-faithful states:
\begin{proposition}[Extension of the metric to $\cD(\cH)$]\label{prop2}
	Let $\rho,\omega\in\cD(\cH)$ and let $\{\rho_n\}_{n\in\NN}$ and $\{\omega_n\}_{n\in\NN}$ be sequences of faithful states satisfying
	\begin{align}\label{conv}
		\tr[(\rho-\rho_n)^2]\to0,~~~~~\tr[(\omega-\omega_n)^2]\to0,
		\end{align}
	as $n\to\infty$. Then the sequence $\{W_{2,\LL}(\rho_n,\omega_n)\}_{n\in\NN}$ converges.
\end{proposition}	

\begin{proof}
	The proof is similar to the one given in Proposition 4.5 of \cite{[CM14]}. It is enough to show that $\{W_{2,\LL}(\rho_n,\omega_n)\}_{n\in\NN}$ is Cauchy. By the triangle inequality, it is even enough to prove that $W_{2,\LL}(\rho_n,\rho_m)\to 0$ as $m,n\to\infty$.	Let $\eps\in (0,1)$ and set $\bar{\rho}:=(1-\eps)\rho+\eps \mathbb{I}$. Let $N\in\NN$ be such that for any $n\ge N$, $\tr[(\rho-\rho_n)^2]\le \eps^2$. For $n\ge N$, consider the convex interpolation $\gamma(s):= (1-s) \rho_n + s\bar{\rho} $. Since $\gamma(s)  \ge \eps s\mathbb{I}$ for $s\in [0,1]$, we find from \Cref{length} that
	\begin{align*}
		W_{2,\LL}&(\rho_n,\bar{\rho})\le \int_0^1   \|\dot{\gamma}(s)\|_{g_{\LL,\gamma(s)}}ds\\
		&=\int_0^1 \left[{ \sum_{j\in\cJ}   c_j \langle  \partial_j (-\operatorname{div}\circ [\gamma(s)]_{\vec{\omega}}\circ\nabla)^{-1} (\dot{\gamma}(s)), [\gamma(s)]_{\omega_j} \partial_j   (-\operatorname{div}\circ [\gamma(s)]_{\vec{\omega}}\circ\nabla)^{-1} (\dot{\gamma}(s)) \rangle   } \right]^{1/2}ds \\
		&=   \int_0^1  \sqrt{\langle  \dot{\gamma}(s),  (   -\operatorname{div}([\gamma(s)]_{\vec{\omega}} \nabla (.)) )^{-1}~ \dot{\gamma}(s)\rangle} ds \\
		&\le \sqrt{\frac{K_\LL}{\eps }}\int_0^1 s^{-1/2} \sqrt{\tr[(\dot{\gamma}(s))^2]}  ds,  
		\end{align*}
	where we used \Cref{lemmametric} in the second, third and fourth above lines. Now 
	\begin{align*}
	\tr[(\dot{\gamma}(s))^2]&=\tr[ (\rho-\rho_n  +\eps(\mathbb{I}   -\rho))^2  ]\\
	&\le 2\tr [  (  \rho-\rho_n   )^2  ]+2\eps^2 \tr[(\mathbb{I}  -\rho)^2]\\
	&\le 2\left(  1+\tr[(\mathbb{I}-\rho)^2]\right)\eps^2.
\end{align*}	
Hence, $W_{2,\LL}(\rho_n,\bar{\rho})\le \sqrt{K(\LL,\rho) \eps } $, for some constant $K(\LL,\rho)$ depending on $\rho$ and $\LL$. Since $\eps$ was arbitrary, we conclude by triangle inequality that $W_{2,\LL}(\rho_m,\rho_n)\le W_{2,\LL}(\rho_m,\bar{\rho})+ W_{2,\LL}(\bar{\rho},\rho_n)   \to 0$.
	\qed
	\end{proof}
\bigskip
\noindent
The above proposition justifies the following definition: The modified Wasserstein distance $W_{2,\LL}$ between two states $\rho,\omega\in\cD(\cH)$ is defined as
\begin{align*}
	W_{2,\LL}(\rho,\omega):=\lim_{n\to\infty} W_{2,\LL}(\rho_n,\omega_n),
\end{align*}	 
where $\{\rho_n\}_{n\in\NN}$ and $\{\omega_n\}_{n\in\NN}$ are arbitrary sequences in $\cD_+(\cH)$ satisfying \reff{conv}. It can be shown that $(\cD(\cH),W_{2,\LL})$ forms a complete metric space as follows:
\begin{lemma}\label{W1W2}
	For any $\rho,\omega\in\cD(\cH)$, 
	\begin{align*}
		\|\rho-\omega\|_1\le  2 \left(  \sum_{j\in\mathcal{J}} c_j (\e^{-\omega_j/2}+\e^{\omega_j/2})\|\tilde{L}_j\|^2_\infty\right)^{1/2} ~W_{2,{\LL}}(\rho,\omega).
	\end{align*}
\end{lemma}
\begin{proof}
	The proof follows from a direct application of inequality (2.39) of Lemma 6 of \cite{rouze2017concentration}: for any $X\in\cB_{sa}(\cH)$:
	\begin{align*}
		|\tr(X(\rho-\omega))|\le \sqrt{d}~\|X\|_{\operatorname{Lip}}W_{2,\LL}(\rho,\omega),
	\end{align*}
where 
	\begin{align*}
	\|X\|_{\operatorname{Lip}}&:= \left(\frac{1}{d}  \sum_{j\in\mathcal{J}} c_j (\e^{-\omega_j/2}+\e^{\omega_j/2})\|\partial_jX\|_{\infty}^2\right)^{1/2}\\
	&\le \frac{2}{\sqrt{d}}\left(  \sum_{j\in\mathcal{J}} c_j (\e^{-\omega_j/2}+\e^{\omega_j/2})\|\tilde{L}_j\|^2_\infty\right)^{1/2}\|X\|_\infty.
\end{align*}
The result follows from the duality relation between the norms $\|.\|_\infty$ and $\|.\|_1$.
	\qed
\end{proof}
\begin{proposition}\label{complete}
	The metric space $(\cD(\cH),W_{2,\LL})$ is complete.
\end{proposition}	
\begin{proof}
	This directly follows from \Cref{W1W2} and \Cref{prop2}: assume that $\{\rho_n\}_{n\in\NN}$ is a Cauchy sequence in $(\cD(\cH),W_{2,\LL})$, that is $W_{2,\LL}(\rho_n,\rho_m)\to 0$ as $m,n\to \infty$. Then, by \Cref{W1W2}, $\{\rho_n\}_{n\in\NN}$ is also Cauchy with respect to the trace norm $\|.\|_1$. By completeness of the normed vector space $(\cB(\cH),\|.\|_1)$, this implies existence of $\rho_\infty\in\cB(\cH)$ such that $\|\rho_n-\rho_\infty\|_1\to0$ as $n\to\infty$. Moreover, $\rho_\infty\in\cD(\cH)$: indeed, for any $\psi\in(\cH,\langle .|.\rangle)$, 
	\begin{align*}
		\langle \psi| \rho_\infty\psi\rangle&=\langle \psi| (\rho_\infty-\rho_n)\psi\rangle+\langle\psi|\rho_n\,\psi\rangle,
	\end{align*}	
 which implies the positivity of $\rho_\infty$, since $|\langle \psi| (\rho_\infty-\rho_n)\psi\rangle|\le\|\rho_\infty-\rho_n\|_1\langle\psi|\psi\rangle\to0$ as $n\to\infty$, and $\langle \psi|\rho_n\psi\rangle\ge 0$ for all $n$. Moreover 
 \begin{align*}
 	|\tr (\rho_n-\rho_\infty)|\le \|\rho_n-\rho_\infty\|_1\to0, ~~~n\to\infty,
 	\end{align*}
 which implies $\tr\rho_\infty=1$. We conclude that $W_{2,\LL}(\rho_n,\rho_\infty)\to W_{2,\LL}(\rho_\infty,\rho_\infty)=0$ by \Cref{prop2}.
\qed
\end{proof}

\subsection{Quantum functional and transportation cost inequalities}\label{2.4}
A primitive QMS $(\Lambda_t)_{t\ge 0}$ with unique invariant state $\sigma$ is said to satisfy:
\begin{itemize}
	\item[1.] a Poincar\'{e} inequality with constant $\lambda>0$,  if for all $f\in\cB_{sa}(\cH)$ with $\tr(\sigma f)=0$,
	\begin{align}\label{p}
		\tag{PI($\lambda$)}
		\lambda	\operatorname{Var}_\sigma(f)^2\le -\tr( \sigma f\LL(f)),
	\end{align}
where $\operatorname{Var}_\sigma(f):=\tr (\sigma f^2)-\tr(\sigma f)^2 $.
	\item[2.] a modified logarithmic Sobolev inequality with constant $\alpha_1>0$ if for all $\rho\in\cD_+(\cH)$,
	\begin{align}\label{ls1}\tag{$\operatorname{MLSI}(\alpha_1)$}
		2	\alpha_1 D(\rho\|\sigma)\le  \operatorname{EP}_\sigma(\rho)=\operatorname{I}_\sigma(\rho).
	\end{align}
\item[3.] a transportation-cost inequality of order $2$ with constant $c_2>0$ if for all $\rho\in \cD_+(\cH)$,
\begin{align}
	\tag{TC$_2$($c_2$)}\label{t2}
	W_{2,\LL}(\rho,\sigma)\le \sqrt{2c_2D(\rho\|\sigma)}.
\end{align}
\item[4.] MLSI+TC$_2$($c$) inequality with constant $c>0$ if for all $\rho\in \cD_+(\cH)$,
\begin{align}\tag{MLSI+TC$_2$($c$)}\label{lst2}
	W_{2,\LL}(\rho,\sigma)\le c\sqrt{\operatorname{I}_\sigma(\rho)}.
\end{align}	
\end{itemize}
That \reff{ls1} implies \reff{t2} for $c_2=\alpha_1^{-1}$ was proved in \cite{rouze2017concentration}. Hence, the following theorem easily follows:
\begin{proposition}
	Assume that $(\Lambda_t)_{t\ge 0}$ satisfies \reff{ls1} for some $\alpha_1>0$. Then it also satisfies \reff{lst2} with $c=\alpha_1^{-1}$.
\end{proposition}

\section{Quantum Ricci lower bound and $\kappa$-displacement convexity}\label{sec3}
In their celebrated paper \cite{[BE85]} (see also \cite{[B06]}), Bakry and Emery found an elegant criterion which implies the logarithmic Sobolev inequality in the setting of diffusions. In this case of Markov semigroups defined on a Riemannian manifold $\mathcal{M}$, this criterion, called the \textit{Ricci lower bound}, which is a special case of the Bakry-Emery condition, was shown later on to be equivalent to the so-called $\kappa$-displacement convexity of the relative entropy along geodesics in the Wasserstein space of probability measures on $\mathcal{M}$ in \cite{[RS05]}. This notion of $\kappa$-displacement convexity was extended to the framework of (necessarily non-diffusive) finite Markov chains by Maas in \cite{maas2011gradient}. Carlen and Maas generalized this notion to the quantum regime in \cite{Carlen20171810} and proved that it implies the modified logarithmic Sobolev inequality as well as the contractivity of the Wasserstein metric under the flow associated to the underlying quantum semigroup $(\Lambda_t)_{t\ge 0}$. In their previous article \cite{[CM14]}, the same authors had already studied this quantum extension of the notion of $\kappa$-displacement convexity in the particular case of the Fermionic Fokker-Planck equation. In this section, we provide a systematic analysis of the $\kappa$-displacement convexity, including a study of the geodesic equations on the Riemannian manifold $(\cD_+(\cH),g_{\LL})$.  
\subsection{Geodesic equations}
Similarly to Theorem 2.4 of \cite{[EM12]}, Carlen and Maas provided in \cite{Carlen20171810} the set of faithful states $\cD_+(\cH)$ with a Riemannian structure with associated Riemannian distance given by $W_{2,\LL}$. Therefore, the local existence and uniqueness of constant speed geodesics is garanteed by standard Riemannian geometry. We first recall that a constant speed geodesic $(\gamma(s),U(s))_{s\in [0,1]}$, where $U$ is related to $\gamma$ through \Cref{continuity2}, satisfies a Euler-Lagrange equation that we derive in \Cref{eulerlagrangeus}. This result is a direct generalization of Theorem 5.3 in \cite{[CM14]}. We start by recalling the abstract framework. Let $(\mathcal{V},\langle.,.\rangle)$ be a finite-dimensional real Hilbert space. Let $\mathcal{W}\subset \mathcal{V}$ be a subspace of $\mathcal{V}$, and $z\in \mathcal{V}\backslash \mathcal{W}$. Consider the affine subspace $\mathcal{W}_z:=z+\mathcal{W}$, and let $\mathcal{M}\subset \mathcal{W}_z$ be a relatively open subset. Let $D:\mathcal{M}\to \cB(\mathcal{W})$ be a smooth function such that $D(x)$ is self-adjoint and invertible for all $x\in \mathcal{M}$. We shall write $C(x):=D(x)^{-1}$. Consider the Lagrangian $L:\mathcal{W}\times \mathcal{M}\to \RR$ defined by $L(p,x)=\langle C(x) p,p\rangle$ and the associated minimization problem:
\begin{align*}
	\inf_{u(.)\in C^1([0,1],\mathcal{M})}\left(\int_0^1 L(u'(t),u(t))dt:~~ u(0)=u_0,~ u(1)=u_1\right),
	\end{align*}
	where $u_0, u_1\in \mathcal{M}$ are given boundary values. Then the Euler-Lagrange equations are equivalent to the following system of equations:
	\begin{align}\label{eulerlagrange}
		\left\{ 
		\begin{aligned}
			&u'(t)-D(u(t))v(t)=0,\\
			& v'(t)+\frac{1}{2}\langle \partial_x D(u(t))v(t),v(t)\rangle=0.
			\end{aligned}
			\right.
			\end{align}
Here, we apply this abstract result to the case where $\mathcal{V}=\cB_{sa}(\cH)$, with inner product $\langle .,. \rangle$ the usual Hilbert-Schmidt inner product, $\mathcal{W}=\{A\in \mathcal{V}:~ \tr(A)=0\}$, $z:=\mathbb{I}/\dim(\cH)$, and $\mathcal{M}=\cD_+(\cH)$.
Indeed any density operator $\rho$ can be written as $\rho=\mathbb{I}/\dim{\cH}+K$, for some self-adjoint and traceless operator $K$. For any $\rho\in\cD_+(\cH)$, we already proved in \Cref{lemmametric} that $D_{\vec{\omega}}(\rho): U\mapsto -\operatorname{div}([\rho]_{\vec{\omega}}\nabla U)$ is invertible and self-adjoint. Now we use the following identity (see \cite{[CM14]} p. 21):
	\begin{align}\label{eq3}
		\left. \frac{d}{dt} (\rho+t A)^{\alpha}\right|_{t=0}=\int_0^1 \int_0^\alpha \frac{\rho^{\alpha -\beta}}{(1-s)\mathbb{I}+s\rho} A \frac{\rho^\beta}{(1-s)\mathbb{I}+s \rho}d\beta ds
		\end{align}
		for any $0<\alpha <1$, $\rho\in \cD_+(\cH)$ and $A \in\mathcal{W}$. Hence for all $A,U\in \mathcal{W}$,
		\begin{align}
			\left. \frac{d}{dt}\right|_{t=0} \langle D_{\vec{\omega}}(\rho+tA)[U],U\rangle&=\left.\frac{d}{dt}\right|_{t=0} \sum_{j\in\cJ} c_j \langle \partial_j U,[\rho+tA]_{\omega_j} \partial_j U\rangle\nonumber\\
			&=\left.\frac{d}{dt}\right|_{t=0}\sum_{j\in\cJ} c_j \langle \partial_j U,\int_0^1\e^{\omega_j(1/2-\alpha)} (\rho+t A)^\alpha \partial_j U(\rho+ tA)^{1-\alpha}\rangle d\alpha\nonumber\\
			&=\langle A,\nabla U._\rho\nabla U\rangle,\label{eq11}
			\end{align}
		where for two vectors $\vec{V}_1, \vec{V}_2$ in $\bigoplus_j \cB(\cH)$,
	\begin{align*}
			\vec{V}_1._\rho \vec{V}_2:= \sum_{j\in\cJ}c_j \int_0^1\int_0^1\e^{\omega_j(1/2-\alpha)}(   \chi_j(\vec{V}_1,\vec{V}_2^*,\rho,\alpha,s)+\chi_j(\vec{V}_1^*,\vec{V}_2,\rho,1-\alpha,s)) d\alpha ds,
			\end{align*}
where
\begin{align*}
	\chi_j(\vec{V}_1,\vec{V}_2,\rho,\alpha,s):= \int_0^\alpha\frac{\rho^\beta}{(1-s)\mathbb{I}+s\rho}(V_1)_j~\rho^{1-\alpha}~(V_2)_j\frac{\rho^{\alpha-\beta}}{(1-s)\mathbb{I}+s\rho}d\beta.
		\end{align*}
			Therefore, in our context the Euler-Lagrange equations \reff{eulerlagrange} reduce to the following:
		\begin{theorem}\label{eulerlagrangeus}
			The geodesic equations in the Riemannian manifold $(\cD_+(\cH),W_{2,\LL})$ are given by
		\begin{align}\label{geod}
	\left\{\begin{aligned}		&\frac{d}{ds}\gamma(s)+\operatorname{div}([\gamma(s)]_{\vec{\omega}} \nabla U(s))=0,\\
			& \frac{dU(s)}{ds}+\frac{1}{2} \nabla U(s)._{\gamma(s)} \nabla U(s)=0.
			\end{aligned}
			\right.
			\end{align}
\end{theorem}		
\subsection{Different formulations of quantum $\kappa$-displacement convexity}
In analogy with \cite{[EM12]}, we say that a primitive quantum Markov semigroup $(\Lambda_t)_{t\ge 0}$ with associated invariant state $\sigma$ and generator $\LL$ of the form of \Cref{LLDBC} has \textit{Ricci curvature bounded from below} by a constant $\kappa\in\RR$ if the following inequality holds:
\begin{align}\tag{$\operatorname{Ric}(\LL)\ge \kappa$}\label{cd}
	\left.	\frac{d^2}{ds^2}\right|_{s=0}D(\gamma(s)\|\sigma)\ge \kappa \| \dot{\gamma}(0)\|_{g_{\LL,\rho}}^2,
\end{align}
where $(\gamma(s),U(s))_{s\in(-\eps,\eps)}$ is the unique solution to the geodesic equation \reff{geod}  such that $\cD_+(\cH)\ni\rho:=\gamma(0)$ and $U(0)=U$. We also refer to the above inequality as the \textit{quantum Ricci lower bound}.
\noindent
\Cref{eulerlagrangeus} is useful to derive an expression for the second derivative of the relative entropy $D(\gamma(s)\|\sigma)$ with respect to $s$, where $(\gamma(s))_{s\in (-\eps,\eps)}$ is a constant speed geodesic with associated tangent vector $\nabla U(s)$ at each $s$. We already know from the gradient flow equation \reff{gradflow} that 
\begin{align*}
	\frac{d}{ds} D(\gamma(s)\|\sigma)&=-g_{\LL,\gamma(s)}(\dot{\gamma}(s),\LL_*(\gamma(s)))\\
	&= \sum_{j\in\cJ} c_j \langle \partial_j U(s),[\gamma(s)]_{\omega_j}\partial_j(\log \gamma(s)-\log\sigma)\rangle\\
	&=\sum_{j\in\cJ} c_j \langle \partial_j U(s),[\gamma(s)]_{\omega_j} (\tilde{L}_j\log(\e^{-\omega_j/2}\gamma(s))-\log(\e^{\omega_j/2}\gamma(s))\tilde{L}_j)\rangle,
	\end{align*}
where the second line comes from Theorem 5.10 in \cite{Carlen20171810}, and the last identity comes from Lemma 5.9 of \cite{Carlen20171810}. Now by identity (5.6) of the same paper,
\begin{align*}
	[\gamma(s)]_{\omega_j} (\tilde{L}_j\log(\e^{-\omega_j/2}\gamma(s))-\log(\e^{\omega_j/2}\gamma(s))\tilde{L}_j)=\e^{-\omega_j/2}\tilde{L}_j\gamma(s)-\e^{\omega_j/2}\gamma(s)\tilde{L}_j,
	\end{align*}
 so that we finally get
 \begin{align*}
 	\frac{d}{ds}D(\gamma(s)\|\sigma)=\sum_{j\in\cJ} c_j \langle \partial_j U(s),\e^{-\omega_j/2}\tilde{L}_j\gamma(s)-\e^{\omega_j/2}\gamma(s)\tilde{L}_j\rangle.
 	\end{align*}
 	Differentiating once more, we get:
 \begin{align}
 \left.	\frac{d^2}{ds^2}D(\gamma(s)\|\sigma)\right|_{s=0}=\sum_{j\in\cJ} c_j ~&\left\{ \langle\partial_j \left.\frac{d}{ds}U(s)\right|_{s=0},\e^{-\omega_j/2}  \tilde{L}_j\rho-\e^{\omega_j/2}\rho \tilde{L}_j\rangle\right.\nonumber\\
 	&\left.+\langle \partial_j U,\e^{-\omega_j/2}\tilde{L}_j \dot{\gamma}(0)-\e^{\omega_j/2}\dot{\gamma}(0) \tilde{L}_j\rangle\right\}.\label{eq12}
 	\end{align}
We first take care of the second line of \Cref{eq12}. Using \Cref{eulerlagrangeus} as well as \Cref{LLDBC}, we find
 \begin{align*}
 	\langle \partial_j U,\e^{-\omega_j/2}\tilde{L}_j \dot{\gamma}(0)&-\e^{\omega_j/2}\dot{\gamma}(0) \tilde{L}_j\rangle\\
 	&=	-\langle \partial_j U,\e^{-\omega_j/2}\tilde{L}_j \operatorname{div}([\rho]_{\vec{\omega}} \nabla U)-\e^{\omega_j/2}\operatorname{div}([\rho]_{\vec{\omega}} \nabla U)\tilde{L}_j\rangle\\
 	&= -\langle \partial_j U,\e^{-\omega_j/2}\tilde{L}_j\sum_{k\in\cJ} c_k[[\rho]_{\omega_k} \partial_k U, \tilde{L}_k^*  ]-\e^{\omega_j/2} \sum_{k\in\cJ} c_k [[\rho]_{\omega_k}  \partial_k U,\tilde{L}_k^*  ]\tilde{L}_j\rangle\\
 	&=\sum_{k\in\cJ} c_k \left( \e^{-\omega_j/2}\langle \partial_k (\tilde{L}_j^* \partial_j U),[\rho]_{\omega_k}\partial_k U\rangle-\e^{\omega_j/2}\langle \partial_k ( \partial_j U \tilde{L}_j^*),[\rho]_{\omega_k}\partial_k U\rangle\right)\\
 	&= \sum_{k\in\cJ}  c_k \langle \partial_k \left( \e^{-\omega_j/2} \tilde{L}_j^* \partial_j U-\e^{\omega_j/2} \partial_j U \tilde{L}_j^*\right),[\rho]_{\omega_k}\partial_k U\rangle.
 	\end{align*}
 Hence by \reff{LLDBC},
 \begin{align}
\sum_{j\in\cJ} c_j 	\langle \partial_j U,\e^{-\omega_j/2}\tilde{L}_j \dot{\gamma}(0)-\e^{\omega_j/2}\dot{\gamma}(0) \tilde{L}_j\rangle&= -\sum_k c_k \langle \partial_k \LL(U),[\rho]_{\omega_k}\partial_k U\rangle\nonumber\\
&=-\langle \nabla \LL(U),\nabla U\rangle_{\LL,\rho}.\label{eq13}
 	\end{align}
 By \reff{geod}, the first line of \reff{eq12} is equal to  
 \begin{align}
 	\frac{1}{2} \sum_{j\in\cJ} c_j \langle \partial_j (\nabla U._{\rho}\nabla U),&\e^{\omega_j/2} \rho\tilde{L}_j-\e^{-\omega_j/2} \tilde{L}_j\rho\rangle\nonumber\\
 	&=\frac{1}{2}  \sum_{j\in\cJ} c_j \langle \nabla U._\rho\nabla U,[\tilde{L}_j^*,\rho \tilde{L}_j]\e^{\omega_j/2}-\e^{-\omega_j/2} [\tilde{L}_j^*,\tilde{L}_j\rho]\rangle\nonumber\\
 	&= \frac{1}{2}  \langle \nabla U._{\rho} \nabla U, \LL_*(\rho)\rangle\label{eq14},
 	\end{align}
 where we used that, replacing $\tilde{L}_j$ by $\tilde{L}_j^*$ so that $\omega_j\rightarrow -\omega_j$ and $c_j\to c_j$,
 \begin{align*}
 	\LL_*(\rho)=\sum_{j\in\cJ} c_j \left(  \e^{\omega_j/2}   [\tilde{L}_j^*\rho,\tilde{L}_j]+\e^{-\omega_j/2}[\tilde{L}_j,\rho \tilde{L}_j^*]    \right)= \sum_{j\in\cJ} c_j \left( \e^{-\omega_j/2}  [\tilde{L}_j\rho,\tilde{L}_j^*]+\e^{\omega_j/2} [\tilde{L}_j^*,\rho \tilde{L}_j]\right),
 \end{align*}
Hence, using \reff{eq13} and \reff{eq14}, \reff{eq12} reduces to
 	\begin{align}\label{secondderiv}
 	\left.\frac{d^2}{ds^2}D(\gamma(s)\|\sigma)\right|_{s=0}&=\frac{1}{2} \langle \nabla U._{\rho}\nabla U,\LL_*(\rho)\rangle	-\langle \nabla\LL(U),\nabla U\rangle_{\LL,\rho}.
 \end{align}	
 	One can compare this expression with the one derived in Proposition 4.3 of \cite{[EM12]}. To make this analogy more clear, we denote the quantity on the right hand side of \Cref{secondderiv} by $B(\rho,U)$ so that
 	\begin{align}\label{BB}
 	\left.	\frac{d^2}{ds^2}D(\gamma(s)\|\sigma)\right|_{s=0}= B(\rho,U).
 		\end{align}
 	The following lemma extends Lemma 4.6 of \cite{[EM12]} to the quantum regime, as well as part of the proof of Proposition 5.11 of \cite{[CM14]}, and is proven to be useful in what follows:
 	\begin{lemma}\label{lemmaimportant}
 		Let $(\gamma(s))_{s\in[0,1]}$ be a smooth curve in $\cD_+(\cH)$. For each $t\ge 0$, set $\gamma(s,t):=\Lambda_{*st}(\gamma(s))$, and let $(U(s,t))_{s\in[0,1]}$ be a smooth curve satisfying the continuity equation
 		\begin{align}
 			\partial_s(\gamma(s,t))+\operatorname{div} ([\gamma(s,t)]_{\vec{\omega}}\nabla U(s,t))=0,~~~~~~~s\in[0,1].\label{eq15}
 			\end{align}
 		Therefore,
 		\begin{align*}
 			\frac{1}{2}\partial_t\| \partial_s\gamma(s,t)\|^2_{g_{\LL,\gamma(s,t)}}+\partial_s D(\gamma(s,t)\|\sigma)=-s B(\gamma(s,t),U(s,t)).
 		\end{align*}	
 	\end{lemma}	
 	\begin{proof}
 		Start by noticing that
 		\begin{align*}
 			\partial_s D(\gamma(s,t)\|\sigma)&=\partial_s \tr(\gamma(s,t) (\log\gamma(s,t)-\log\sigma))\\
 			&=\tr (\partial_s \gamma(s,t)  (\log\gamma(s,t)-  \log\sigma))\\
 			&=-\tr( (\log \gamma(s,t)  -\log\sigma) \operatorname{div} ([\gamma(s,t)]_{\vec{\omega}}   \nabla U(s,t)) )\\
 			&= -\langle \log \gamma(s,t) -\log\sigma ,\operatorname{div} ([\gamma(s,t)]_{\vec{\omega}}  \nabla U(s,t))\rangle\\
 			&=-\sum_{j\in\cJ} c_j \langle \log\gamma(s,t)-\log\sigma,[   [\gamma(s,t)]_{\omega_j}(\partial_j U(s,t)) ,\tilde{L}_j^*]\rangle\\
 			&=\sum_{j\in\cJ} c_j\langle \partial_j (\log\gamma(s,t)-\log\sigma),[\gamma(s,t)]_{\omega_j}(\partial_jU(s,t))\rangle\\
 			&=\sum_{j\in\cJ} c_j \langle [\gamma(s,t)]_{\omega_j}(\partial_j   (\log\gamma(s,t) -\log\sigma )),\partial_j U(s,t)\rangle\\
 			&=\sum_{j\in\cJ} c_j \langle \partial_j^*[\gamma(s,t)]_{\omega_j}(\partial_j (\log\gamma(s,t)-\log\sigma)),U(s,t)\rangle\\
 			&=-\langle \LL_*(\gamma(s,t)),U(s,t)\rangle,
 		\end{align*}
 		where in the third line we used \reff{eq15}, in last line we used Theorem 5.10 of \cite{Carlen20171810}, and in the second line we used that
 		\begin{align*}
 	\tr(\gamma(s,t) \partial_s \log\gamma(s,t)  )&=  \tr \left(   \gamma(s,t)  \partial_s \int_0^\infty  \frac{1}{(1+u)\mathbb{I}} -\frac{1}{\gamma(s,t)+u\mathbb{I}} ~du   \right)\\
 	&= \tr \gamma(s,t)  \int_0^\infty \frac{1}{\gamma(s,t) +u \mathbb{I}}\partial_s \gamma(s,t) \frac{1}{\gamma(s,t) +u \mathbb{I}}~du\\
 	&= \tr \partial_s \gamma(s,t)\\
 	&=0.
 		\end{align*}
 		 Moreover, by definition of the metric $g_{\LL}$ through \Cref{gammav,innerj},
 		\begin{align}
 			\frac{1}{2} \partial_t \|\partial_s \gamma(s,t)\|_{g_{\LL,\gamma(s,t)}}^2&=\frac{1}{2}\partial_t\sum_{j\in\cJ}c_j \langle \partial_j U(s,t),[\gamma(s,t)]_{\omega_j} \partial_j U(s,t)\rangle\nonumber\\
 			&=\sum_{j\in\cJ} c_j\Big( \langle \partial_t (\partial_j U(s,t)),[\gamma(s,t)]_{\omega_j}\partial_j U(s,t)\rangle\nonumber\\
 			&~~~~~~+\frac{1}{2}\langle \partial_j U(s,t),\partial_t ([\gamma(s,t)]_{\omega_j}  )\partial_j U(s,t)\rangle\Big).\label{eq16}
 		\end{align}
 		From \reff{eq11}, 
 		\begin{align}
 			\sum_{j\in\cJ} c_j \langle \partial_j U(s,t),\partial_t( [\gamma(s,t)]_{\omega_j})\partial_j U(s,t)\rangle&= \langle \partial_t \gamma(s,t),\nabla U(s,t)._{\gamma(s,t)} \nabla U(s,t)\rangle\nonumber\\
 			&=s \langle \LL_*(\gamma(s,t)),\nabla U(s,t)._{\gamma(s,t)} \nabla U(s,t)\rangle.\label{eq17}
 		\end{align}	
 	Moreover,
 	\begin{align}
 		\sum_{j\in\cJ} c_j \langle \partial_t (\partial_j U(s,t)), [\gamma(s,t)]_{\omega_j}\partial_j U(s,t)\rangle&=-\sum_{j\in\cJ} c_j \langle \partial_t U(s,t), [ [\gamma(s,t)]_{\omega_j} (\partial_j U(s,t)),\tilde{L}_j^*]\rangle\nonumber\\
 		&=		-\langle \partial_t U(s,t),\operatorname{div}([\gamma(s,t)]_{\vec{\omega}}(\nabla U(s,t))  )\rangle\nonumber\\
 		&= \langle \partial_t U(s,t),\partial_s \gamma(s,t)\rangle\nonumber\\
 		&= \partial_t (\langle U(s,t),\partial_s \gamma(s,t)\rangle ) -\langle U(s,t),\partial_s\partial_t \gamma(s,t)\rangle\nonumber\\
 		&= \partial_t\|\partial_s\gamma(s,t)\|_{g_{\LL,\gamma(s,t)}}^2-\langle U(s,t),\partial_s\partial_t \gamma(s,t)\rangle\nonumber\\
 		&=\partial_t\|\partial_s\gamma(s,t)\|_{g_{\LL,\gamma(s,t)}}^2- \langle U(s,t),\partial_s(s\LL_*(\gamma(s,t)))\rangle,\label{eq18}
 	\end{align}	
 where we used once again \reff{eq15} in the third and fifth lines above.
 Therefore, using \reff{eq17} and \reff{eq18}, the right hand side of \reff{eq16} reduces to
 \begin{align*}
 	\frac{1}{2}\partial_t \|  \partial_s\gamma(s,t)\|_{g_{\LL,\gamma(s,t)}}^2=\langle U(s,t),\partial_s(s\LL_*(\gamma(s,t)))\rangle-\frac{1}{2}s\langle \LL_*(\gamma(s,t)),\nabla U(s,t)._{\gamma(s,t)}\nabla U(s,t)\rangle.
 	\end{align*}
 Hence,
 \begin{align*}
 		\frac{1}{2}\partial_t \| & \partial_s\gamma(s,t)\|_{g_{\LL,\gamma(s,t)}}^2 +\partial_s D(\gamma(s,t)\|\sigma)\\
 		&=s \langle U(s,t),\LL_*\partial_s (\gamma(s,t))\rangle -\frac{1}{2}s \langle \LL_*\gamma(s,t).\nabla U(s,t)._{\gamma(s,t)}\nabla U(s,t)\rangle  \\
 		&= -s \langle \LL(U(s,t)),\operatorname{div}([\gamma(s,t)]_{\vec{\omega}} \nabla U(s,t)   )\rangle -\frac{1}{2} s \langle \LL_*\gamma(s,t).\nabla U(s,t)._{\gamma(s,t)}\nabla U(s,t)\rangle  \\
 		&=s \langle \nabla \LL(U(s,t)),\nabla U(s,t)\rangle_{\LL,\gamma(s,t)}-\frac{1}{2}s\langle \LL_*(\gamma(s,t)),\nabla U(s,t)._{\gamma(s,t)}\nabla U(s,t)\rangle\\
 		&=-s B(\gamma(s,t),U(s,t)),
 		\end{align*}
 	which is what needed to be proved.
 		\qed
 	\end{proof}
 
 \begin{theorem}\label{theofund}
 	Let $\LL$ be the generator of an ergodic QMS $(\Lambda_t)_{t\ge 0}$, with unique invariant state $\sigma$, of the form of \Cref{LLDBC}. Then, for $\kappa\in\RR$, the following holds: $(i)\Leftrightarrow (ii)\Rightarrow (iii)\Rightarrow (iv)\Rightarrow (v)$, where
 	\begin{itemize}
 		\item[(i)] $\operatorname{Ric}(\LL)\ge \kappa$
 		\item[(ii)] For all $\rho\in\cD_+(\cH)$, and $U\in\mathcal{W}$,
 		\begin{align*}
 			B(\rho, U)\ge \kappa \|\nabla U\|_{\LL,\rho}^2
 			\end{align*}
 		\item[(iii)] For all $\rho,\omega\in\cD_+(\cH)$ and all $t\ge 0$, writing $\rho_t:=\Lambda_{*t}(\rho)$:
 		\begin{align}
 			\frac{1}{2}\left.\frac{d}{dt}\right|_{t^+} \left(W_{2,\LL}(\rho_t,\omega)\right)^2+\frac{\kappa}{2} W_{2,\LL}(\rho_t,\omega)^2\le D(\omega\| \sigma)-D(\rho_t\|\sigma).\label{star}
 		\end{align}
 	\item[(iv)] \Cref{star} holds for any $\rho,\omega\in\cD(\cH)$.
 	\item[(v)] $\kappa$-displacement convexity of the relative entropy: for any constant speed geodesic $(\gamma(s))_{s\in [0,1]}$ in $\cD(\cH)$,
 	\begin{align*}
 		D(\gamma(s)\|\sigma)\le (1-s)D(\gamma(0)\|\sigma)+s D(\gamma(1)\|\sigma)-\frac{\kappa}{2}s(1-s) W_{2,\LL}(\gamma(0),\gamma(1))^2.
 	\end{align*}	
 	\end{itemize}	
 \end{theorem}	
 \begin{proof}
 	The proof is inspired by the one of Theorem 4.5 of \cite{[EM12]}
That $(i)\Leftrightarrow (ii)$ follows from \Cref{BB}. We use \Cref{lemmaimportant} to show that $(ii)\Rightarrow(iii)$: Take a smooth path $(\gamma(s),U(s))_{s\in[0,1]}$ such that $\gamma(0)=\omega$, $\gamma(1)=\rho$ and 
\begin{align}\label{epseps}
	\int_0^1 \|\dot{\gamma}(s)\|_{g_{\LL,\gamma(s)}}^2ds\le W_{2,\LL}(\rho,\omega)^2+\eps.
\end{align}
With the notations of \Cref{lemmaimportant}, 
\begin{align*}
\frac{1}{2}\partial_t \left( \e^{2\kappa s t}\|\partial_s 
\gamma(s,t) \|_{g_{\LL,\gamma(s,t)}}^2    \right)+\partial_s \left(  \e^{2\kappa s t}D(\gamma(s,t)\|\sigma)  \right)\le 2\kappa t\e^{2\kappa s t }D(\gamma(s,t)\|\sigma).
\end{align*}	
Integrating with respect to $t\in [0,h]$, for some $h>0$, and $s\in[0,1]$,
\begin{align}
	\frac{1}{2}\int_0^1 \left( \e^{2\kappa s h}\|\partial_s \gamma(s,h) \|_{g_{\LL,\gamma(s,h)}}^2  -\|\partial_s \gamma(s,0)\|_{g_{\LL,\gamma(s,0)}}^2   \right)ds +&\int_0^h \left(   \e^{2\kappa t} D(\gamma(1,t)\|\sigma)- D(\gamma(0,t)\|\sigma)    \right)dt\nonumber\\
	&\le 2\kappa \int_0^1ds\int_o^h dt~ t\,\e^{2\kappa st} D(\gamma(s,t)\|\sigma).\label{eq2}
\end{align}	
The following inequality, for which a classical equivalent is given in the proof of Theorem 4.5 of \cite{[EM12]}, can be derived similarly to Lemma 5.1 of \cite{daneri2008eulerian}:
\begin{align}\label{eq21}
	m(\kappa h)W_{2,\LL}(\rho_h,\omega)^2\le \int_0^1 \e^{2\kappa s h} \|\partial_s \gamma(s,h) \|_{g_{\LL,\gamma(s,h)}}^2 ds,
	\end{align}
where $m(x):= x\e^x/\sinh(x)$. Indeed, define $f:s\mapsto \e^{2\kappa s h}$, and denote $L_{f}:=\int_0^1 \frac{1}{f(s)}ds$. Then, let $g:[0,1]\mapsto [0,1]$ be the smooth increasing map defined as $g(s)=L_{f}^{-1}\int_0^s\frac{1}{f(u)}du$, and denote its inverse $k$ such that $k'(g(s))=L_{f}f(s).$ Then define the reparametrized curve $(\gamma(k(r),h ),\,k'(r)\,U(k(r),h) )_{r\in[0,1]}$ which satisfies the continuity equation:
\begin{align*}
	\partial_r \gamma(k(r),h)&=k'(r) \partial_1 \gamma(k(r),h)\\
	&=-k'(r) \operatorname{div}([\gamma(k(r),h)]_{\vec{\omega}}\nabla U(k(r),h)  ),
\end{align*}	
where we used \Cref{eq15} in order to established the second line. This curve satisfies $\gamma(k(0),h)=\omega$ and $\gamma(k(1),h)=\rho_h$, so that
\begin{align*}
	W_{2,\LL}(\rho_h,\omega)^2&\le \int_0^1 \|\partial_r \gamma(k(r),h)\|^2_{g_{\LL,\gamma(k(r),h)}}dr\\
	&= \int_0^1 k'(r)^2 \|\nabla U(k(r),h)\|^2_{\LL,\gamma(k(r),h)}dr\\
	&= \int_0^1 k'(g(s)) \|\nabla U(s,h)\|_{\LL,\gamma(s,h)}^2ds\\
	&= L_f \int_0^1 f(s)\|\partial_s \gamma(s,h)\|_{g_{\LL,\gamma(s,h)}}^2ds,
\end{align*}	
which directly leads to \reff{eq21}. This inequality, together with \reff{epseps}, implies
\begin{align*}
	\frac{m(h\kappa)}{2}&W_{2,\LL}(\rho_h,\omega)^2-\frac{1}{2}W_{2,\LL}(\rho,\omega)^2-\eps+\int_0^h\e^{2\kappa t} dt~D(\rho_h\|\sigma)-hD(\omega\|\sigma)\\
	&\le \frac{1}{2}\int_0^1
	 \e^{2\kappa s h}\|\partial_s \gamma(s,t)\|_{g_{\LL,\gamma(s,t)}}^2 ds-\frac{1}{2}\int_0^1 \|\dot{\gamma}(s)\|_{g_{\LL,\gamma(s)}}^2ds +\int_0^h \e^{2\kappa t} D(\rho_t\|\sigma)~dt-hD(\omega\|\sigma)\\
	 &\le  2\kappa \int_0^1ds\int_0^h dt ~t\e^{2\kappa st }D(\gamma(s,t)\|\sigma).
\end{align*}	
where, in the first inequality, we also used the monotonicity of the relative entropy so that $D(\rho_h\|\sigma)=D(\rho_h\|\Lambda_{*h}\sigma)\le D(\rho_t\|\Lambda_{*t}\sigma)=D(\rho_t\|\sigma)$, and in the second one that for all $t>0$, $\gamma(1,t)=\rho_t$, $\gamma(0,t)=\omega$, as well as \reff{eq2}. Since for all $s\in[0,1]$, $t\mapsto D(\gamma(s,t)\|\sigma)$ is bounded, 
\begin{align*}
	\lim_{h\to 0} \frac{1}{h}\int_0^1\int_0^h ~t\e^{2\kappa st} D(\gamma(s,t)\|\sigma)dt~ds=0.
	\end{align*}
Moreover, 
\begin{align*}
	\lim_{h\to 0}\frac{1}{h}\left(  \int_0^h \e^{2\kappa t}dt~D(\rho_h\|\sigma)-hD(\omega\|\sigma)\right)=D(\rho\|\sigma)-D(\omega\|\sigma)
	\end{align*}
Since $\eps>0$ was arbitrary, we arrive at
\begin{align*}
	\left.\frac{d}{dh}\right|_{h=0^+}\left( \frac{m(\kappa h)}{2}W_{2,\LL}(\rho_h,\omega)^2\right) +D(\rho\|\sigma)-D(\omega\|\sigma)\le 0.
	\end{align*}
	The result for $t=0$ follows from the fact that the first term in the left hand side above is equal to $\frac{\kappa}{2} W_{2,\LL}(\rho_h,\omega)^2+\frac{1}{2}\left.\frac{d}{dh}\right|_{h=0^+}W_{2,\LL}(\rho_h,\omega)^2$. The case $t\ge 0$ directly follows from the case $t=0$.\\\\
	$(iii)\Rightarrow(iv)$ follows from Theorem 3.3 of \cite{daneri2008eulerian} together with the fact that $(\cD(\cH),W_{2,\LL})$ is complete (cf. \Cref{complete}).\\\\
	$(iv)\Rightarrow(v)$ follows directly from Theorem 3.2 of \cite{daneri2008eulerian}.
\qed 	
 \end{proof}

 		\section{A quantum HWI inequality}\label{333}
 		In \cite{Carlen20171810} it was proved that, in the case when $\kappa>0$, \ref{cd} implies \ref{ls1} for $\kappa=\alpha_1$. This is for example the case of the classical and quantum Ornstein-Uhlenbeck processes. Here, we study the case of $\kappa\in\RR$. In \cite{[EM12]}, the authors proved that, in the classical discrete framework, \ref{cd} for $\kappa\in\RR$ implies an HWI-like inequality (see Theorem 7.3). Here, we provide a quantum generalization of their result. 
 		\begin{theorem}\label{theHWI}
 			Assume that $\operatorname{Ric(\LL)}\ge \kappa$, for some $\kappa\in\RR$. Then $\LL$ satisfies the following inequality
 			\begin{align}\tag{$\operatorname{HWI}(\kappa)$}\label{HWI}
 			\forall \rho\in\cD_+(\cH),~~~~~~~	D(\rho\|\sigma)\le W_{2,\LL}(\rho,\sigma) \sqrt{\operatorname{I}_\sigma(\rho)} -\frac{\kappa}{2}W_{2,\LL}(\rho,\sigma)^2.
 				\end{align}
 			
 		\end{theorem}	
\begin{proof}
	By \Cref{theofund}, for any $\rho,\omega\in\cD_+(\cH)$
	\begin{align*}
		\frac{1}{2}\left.\frac{d}{dt}\right|_{t=0^+}W_{2,\LL}(\rho_t,\omega)^2+\frac{\kappa}{2}W_{2,\LL}(\rho,\omega)^2\le D(\omega\|\sigma)-D(\rho\|\sigma).
		\end{align*}
	Taking $\omega:=\sigma$, this implies that
	\begin{align}\label{eq1}
		D(\rho\|\sigma)\le -\frac{1}{2}\left. \frac{d}{dt}\right|_{t=0^+} W_{2,\LL}(\rho_t\|\sigma)^2-\frac{\kappa}{2}W_{2,\LL}(\rho,\sigma)^2.
		\end{align}
	Then, 
	\begin{align*}
		-\frac{1}{2}\left. \frac{d}{dt}\right|_{t=0^+} W_{2,\LL}(\rho_t,\sigma)^2&=\liminf_{s\to 0^+}\frac{1}{2s}(W_{2,\LL}(\rho,\sigma)^2 -W_{2,\LL}(\rho_s,\sigma)^2 )\\
		&\le \limsup_{s\to 0^+}\frac{1}{2s}(W_{2,\LL}(\rho,\rho_s)^2 +2W_{2,\LL}(\rho,\rho_s)W_{2,\LL}(\rho_s,\sigma)  )\\
		&\le \limsup_{s\to 0^+}\frac{1}{2s}W_{2,\LL}(\rho,\rho_s)^2+W_{2,\LL}(\rho,\sigma)\sqrt{\operatorname{I}_\sigma(\rho)}\\
		&=W_{2,\LL}(\rho,\sigma)\sqrt{\operatorname{I}_\sigma(\rho)}.
	\end{align*}
	where the second inequality follows from Lemma 7 of \cite{rouze2017concentration}. The result follows from inserting this back into \reff{eq1}.
	\qed
\end{proof}	
\bigskip
	\noindent
	In the case when $\kappa>0$, we recover the result of $\cite{Carlen20171810}:$
	\begin{corollary}[Quantum Bakry-\'{E}mery theorem]\label{theo5}
		Assume that $\operatorname{Ric}(\LL)\ge \kappa$, for some $\kappa>0$. Then $\LL$ satisfies $\operatorname{MLSI}(\kappa)$.
	\end{corollary}	
\begin{proof}
By \Cref{theHWI}, $\LL$ satisfies \ref{HWI}. $\operatorname{MLSI}(\kappa)$ follows from an application of Young's inequality:
\begin{align}\label{Young}
	xy\le cx^2+\frac{1}{4c}y^2,~~~~~\forall x,y\in\RR, c>0,
\end{align}	
	in which we set $x=W_{2,\LL}(\rho,\sigma)$, $y=\sqrt{\operatorname{I}_\sigma(\rho)}$, and $c=\frac{\kappa}{2}$.
	\qed
\end{proof}	
\bigskip
\noindent
In the case $\operatorname{Ric}(\LL)\ge \kappa$ for $\kappa\in\RR$, HWI$(\kappa)$ still implies a modified log-Sobolev inequality under the further condition that a transportation cost inequality holds. This is a direct quantum generalization of Theorem 7.8 of \cite{[EM12]} (see also Corollary 3.1 of \cite{[OV00]})
\begin{corollary} \label{cor1}  Assume that $\operatorname{Ric}(\LL)\ge \kappa$, $\kappa\in\RR$, and that $\operatorname{TC}_2(c_2)$ holds with $c_2^{-1}\ge \max( 0,-\kappa)$. then $\operatorname{MLSI}(\alpha_1)$ holds for
	\begin{align*}
	\alpha_1=\max\left[    \frac{1}{4c_2}\left( 1+c_2\kappa\right)^2 ,\,\kappa \right]
	\end{align*}	
\end{corollary}	
\begin{proof}
The proof is identical to the one of Corollary 3.1 of \cite{[OV00]}.
	\qed
\end{proof}	
\bigskip
\noindent
Similarly, we can show that $\operatorname{Ric}(\LL)\ge \kappa$ for $\kappa\in\RR$ implies MLSI as long as MLSI+TC$_2$ holds.
\begin{corollary}
	Assume that $\operatorname{Ric}(\LL)\ge \kappa$, $\kappa\in\RR$, and that the inequality $\operatorname{MLSI}+\operatorname{TC}_2(c)$ (defined in \Cref{2.4}) holds with $c^{-1}\ge \max (\kappa,0)$, then $\operatorname{MLSI}(\alpha_1)$ holds, with
	\begin{align*}
		\alpha_1=\frac{1}{c\,(2-\kappa c)}.
	\end{align*}	
\end{corollary}
\begin{proof}
	See Corollary 3.2 of \cite{[OV00]}.
\qed
\end{proof}	
 \bigskip
 \noindent
The \textit{diameter} of $\cD(\cH)$ in the Wasserstein distance $W_{2,\LL}$ is defined as follows:
\begin{align*}
	\operatorname{Diam}_{\LL}(\cD(\cH))=\sup_{\rho,\sigma\in\cD(\cH)} W_{2,\LL}(\rho,\sigma).
\end{align*}
Another straightforward consequence of the $\kappa$-displacement convexity of the quantum relative entropy for $\kappa>0$ is the following estimate on the diameter of $\operatorname{Diam}_{\LL}(\cD(\cH))$, which is a quantum analogue of the Bonnet-Myers theorem (see Proposition 7.3 of \cite{erbar2016poincar}). 
\begin{proposition}
Assume that \ref{cd} holds for $\kappa>0$. Then for any two states $\rho,\omega\in\cD(\cH)$,
\begin{align*}
	W_{2,\LL}(\rho,\omega)^2\le \frac{4}{\kappa} (D(\rho\|\sigma)+D(\omega\|\sigma)).
\end{align*}	
Therefore,
\begin{align*}
  \operatorname{Diam}_{\LL}(\cD(\cH))\le  \sup_{\rho,\omega\in\cD(\cH)}\sqrt{\frac{4}{\kappa} (D(\rho\|\sigma)+D(\omega\|\sigma))}.
\end{align*}
\end{proposition}
\begin{proof}
	The result follows directly from the convexity of the quantum relative entropy (cf. (v) of \Cref{theofund}):
	\begin{align*}
	0\le	D(\gamma(1/2)\|\sigma)\le \frac{1}{2}D(\rho\|\sigma)+\frac{1}{2} D(\omega\|\sigma)-\frac{\kappa}{8} W_{2,\LL}(\rho,\omega)^2.
	\end{align*}	
for a given constant speed geodesic $(\gamma(s))_{s\in[0,1]}$ relating $\rho$ and $\omega$.
	\qed
	\end{proof}
\section{From Ricci lower bound to the Poincar\'{e} inequality}\label{poinc}
In this section we show that $\operatorname{Ric}(\LL)\ge 0$ together with a condition of finiteness of the diameter of $\cD(\cH)$ with respect to the distance $W_{2,\LL}$ implies the Poincar\'{e} inequality, hence extending Proposition 5.9 of \cite{erbar2016poincar}. In order to show this, we first need to extend their theorems 3.1 and 3.5 to the quantum regime. Throughout this section, we fix $(\Lambda_t)_{t\ge 0}$ to be a primitive QMS, with unique invariant state $\sigma$ and associated generator $\LL$, satisfying the detailed balance condition. 
\begin{proposition}[Gradient estimate]\label{theo1} \ref{cd} implies the following gradient estimate: for any $\rho\in \cD_+(\cH)$, any $U\in \cB_{sa}(\cH)$ with $\tr(U)=0$, and all $t>0$:
	\begin{align*}
		\|\nabla (\Lambda_t(U))\|_{\LL,\rho}^2\le \e^{-2\kappa t} 	\|\nabla U\|_{\LL,\Lambda_{*t}(\rho)}^2.
	\end{align*}
\end{proposition}	

\begin{proof}
	Define for $u\in [0,t]$ $\rho_u\equiv  \Lambda_{*u}(\rho) $ and $U_{u}\equiv \Lambda_{u}(U)$. Then,
	\begin{align*}
		\Phi(s):=\e^{-2\kappa s}\|\nabla  U_{t-s}\|_{\LL,\rho_s}^2\equiv \e^{-2\kappa s}\sum_{j\in\cJ} c_j \langle \partial_j( U_{t-s}),[\rho_s]_{\omega_j}   \partial_j( U_{t-s})\rangle.
		\end{align*}
	Then, $\Phi(0)  =  \|\nabla  U_t\|_{\LL,\rho}^2$ and $\Phi(t)= \e^{-2\kappa t}\|\nabla  U\|_{\LL,\rho_t}^2  $. It is then enough to prove that $\Phi$ has non-negative derivative to prove the claim. But:
	\begin{align*}
		&\Phi'(s)=2\e^{-2\kappa s}  \left[ -\kappa   \|\nabla  U_{t-s}\|_{\LL,\rho_s}^2 + \frac{1}{2}\frac{\partial}{\partial s}     \|\nabla  U_{t-s}\|_{\LL,\rho_s}^2  \right]\\
		&=2\e^{-2\kappa s}   \left[       -\kappa   \|\nabla  U_{t-s}\|_{\LL,\rho_s}^2 + \sum_{j\in \cJ} c_j \langle \partial_j\partial_s U_{t-s},[\rho_s]_{\omega_j}\partial_jU_{t-s}\rangle  +\frac{1}{2} \sum_{j\in\cJ} c_j\langle \partial_j U_{t-s},\partial_s([\rho_s]_{\omega_j})\partial_jU_{t-s}\rangle \right]\\
		&=2\e^{-2\kappa s}  \left[ -\kappa   \|\nabla  U_{t-s}\|_{\LL,\rho_s}^2 - \langle \nabla \LL (U_{t-s}),\nabla U_{t-s}\rangle_{\LL,\rho_s}+ \frac{1}{2}  \langle \LL_*(\rho_s),\nabla U_{t-s}._{\rho_s}\nabla U_{t-s} \rangle \right]\\
		&=2\e^{-2\kappa s}  \left[ -\kappa   \|\nabla  U_{t-s}\|_{\LL,\rho_s}^2+  B(\rho_s,U_{t-s})\right]
	\end{align*}
where we used \Cref{eq11} in the second line. We conclude by a use of (ii) of \Cref{theofund}.
	\qed
\end{proof}	
\noindent
\begin{proposition}[Reverse quantum Poincar\'{e} inequality]
	Assume that \ref{cd} holds. Then for any $\rho\in\cD_+(\cH)$, any $U\in\cB_{sa}(\cH)$ and all $t>0$:
	\begin{align}\label{reversepoinc}
		\tr(\Lambda_{*t}(\rho) U^2)-\tr(\rho (\Lambda_t(U))^2)\ge \frac{\e^{2\kappa t}-1}{2\kappa} \| \nabla \Lambda_t(U) \|_{\LL,\rho}^2
	\end{align}	
	\end{proposition}
\begin{proof}
	The proof is similar to the one of Theorem 3.5 of \cite{erbar2016poincar}. For $u\ge 0$, let $\rho_u\equiv\Lambda_{*u}(\rho)$ and $U_u\equiv \Lambda_u (U)$. Then, from \Cref{theo1}, 
	\begin{align*}
		&	\e^{2\kappa s} 	\|\nabla U_t\|_{\LL,\rho}^2= 			\e^{2\kappa s} 	\|\nabla (\Lambda_{s}U_{t-s})\|_{\LL,\rho}^2 \\
		&\le 	\|\nabla U_{t-s}\|_{\LL,\rho_s}^2\\
		&= \sum_{j\in\cJ} c_j \langle  \partial_j U_{t-s},[\rho_s]_{\omega_j}\partial_j U_{t-s}\rangle\\
		&\le \sum_{j\in\cJ} c_j \langle  \partial_j U_{t-s},    ( \e^{-\omega_j/2}  R_{\rho_s}   +\e^{\omega_j/2}   L_{\rho_s})\partial_j U_{t-s}\rangle\\
		&=   \sum_{j\in\cJ} c_j \left(  \e^{\omega_j/2}\tr [\rho_s     \partial_j U_{t-s}  (\partial_j U_{t-s}  )^*]   +\e^{-\omega_j/2}\tr[\rho_s (\partial_j U_{t-s}  )^*\partial_j U_{t-s}     ] \right)    \\
		&= \sum_{j\in\cJ} c_j  \left(  \e^{\omega_j/2}  \tr(\rho_s   [\tilde{L}_j,U_{t-s}][U_{t-s},\tilde{L}_j^*])    +  \e^{-\omega_j/2} \tr(\rho_s   [U_{t-s},\tilde{L}_j^*][\tilde{L}_j,U_{t-s}])   \right)  \\
		&=   \sum_{j\in\cJ} c_j \left( \tr(\rho_s   (\tilde{L}_j U_{t-s}^2 \tilde{L}_j^*    \e^{\omega_j/2}      + \e^{-\omega_j/2}   \tilde{L}_j^* U_{t-s}^2 \tilde{L}_j ) ) \right.\\
		&~~~~~\left.+\e^{\omega_j/2}  \tr(\rho_s  (-\tilde{L}_jU_{t-s}\tilde{L}_j^*   U_{t-s}  -U_{t-s} \tilde{L}_j U_{t-s} \tilde{L}_j^* + U_{t-s}  \tilde{L}_j\tilde{L}_j^* U_{t-s}  ) )     \right. \\
	&~~~~~\left. +   \e^{-\omega_j/2}   \tr(\rho_s    (U_{t-s} \tilde{L}_j^*\tilde{L}_j U_{t-s}   -\tilde{L}_j^* U_{t-s}  \tilde{L}_j U_{t-s}  -U_{t-s} \tilde{L}_j^*U_{t-s}    \tilde{L}_j   ) )	   \right)
	\end{align*}
\begin{align*}
		&=\tr (\rho_s\LL(U_{t-s}^2))   +\sum_{j\in \cJ}  c_j  \left(       \e^{\omega_j/2} \tr(\rho_s  (-\tilde{L}_jU_{t-s}\tilde{L}_j^* U_{t-s}    -U_{t-s} \tilde{L}_j U_{t-s} \tilde{L}_j^* + U_{t-s}  \tilde{L}_j\tilde{L}_j^* U_{t-s}  +U_{t-s}^2  \tilde{L}_j  \tilde{L}_j^* ) )     \right.     \\
		&~~~~~~~~~~~~~~~~~~~~~~~~~~~+\left.  \e^{-\omega_j/2} \tr(\rho_s    (U_{t-s} \tilde{L}_j^*\tilde{L}_j U_{t-s}   -\tilde{L}_j^* U_{t-s}  \tilde{L}_j U_{t-s}  -U_{t-s} \tilde{L}_j^* U_{t-s}   \tilde{L}_j +\tilde{L}_j^* \tilde{L}_j U_{t-s}^2     ) )    \right)\\
		&=   \tr(\rho_s \LL(U_{t-s}^2))    -\tr (\rho_s  U_{t-s}\LL(U_{t-s})   )-\tr(\rho_s \LL(U_{t-s}) U_{t-s} ) \\
		&=\frac{\partial}{\partial s}\tr(\rho_s U^2_{t-s})
	\end{align*}	
	where we used (2.38) of \cite{rouze2017concentration}, with $R_\rho(A)\equiv A\rho$ and $L_{\rho}(A)\equiv\rho A$, in the fourth line. The claim follows after integrating the above inequality from $0$ to $t$. 
	\qed
\end{proof}	
\noindent
\begin{theorem}\label{poinca}$\operatorname{Ric}(\LL)\ge 0 ~+~\operatorname{Diam}_{\LL}(\cD(\cH))\le D$ $\Rightarrow$ $\operatorname{PI}(\frac{1}{2\e D^2}).$\label{poinc0}
\end{theorem}
\begin{proof}
	Let $f\in\cB_{sa}(\cH)$ an eigenvector of $\LL$ with associated eigenvalue opposite to the spectral gap $\lambda$ of $\LL$. Without loss of generality, $\|f\|_\infty=1$, and by primitivity of $(\Lambda_t)_{t\ge 0} $, $\tr(\sigma f)=0$. Now, note that $\Lambda_t(f)=\e^{-\lambda t}f$. Therefore, the reverse Poincar\'{e} inequality \reff{reversepoinc} in the case when $\kappa=0$ implies that for any $\rho\in\cD_+(\cH)$,
	\begin{align*}
		\|\nabla f\|_{\LL,\rho}^2\le   \frac{\e^{2\lambda t}}{t}\|f\|_\infty^2.
	\end{align*}
Optimizing in $t$ and using $\|f\|_\infty=1$, we find
\begin{align*}
		\|\nabla f\|_{\LL,\rho}^2\le 2\e \lambda\|f\|_\infty^2.
\end{align*}	
Given the following spectral decomposition of $f=\sum_\mu  \mu P_\mu$, since $\tr(\sigma f)=0$, the minimum and maximum eigenvalues of $f$, respectively denoted by $\mu_{\min}$ and $\mu_{\max}$, obey $\mu_{\min} <0<\mu_{\max}$. Since we assumed $\|f\|_\infty=1$, this implies that given a path $(\gamma(s),U(s))_{s\in[0,1]}$ in $\cD(\cH)$
joining the states $\gamma(0)=\frac{P_{\mu_{\max}}}{\tr(P_{\mu_{\max}})}$ and $\gamma(1)= \frac{P_{\mu_{\min}}}{\tr(P_{\mu_{\min}} )}$ such that $\int_{0}^1  \|\dot{\gamma}(s)\|_{g_{\LL,\gamma(s)}}^2 ds\le W_{2,\LL}(\gamma(0),\gamma(1))^2+\eps$,
\begin{align*}
	1\le |\mu_{\max} - \mu_{\min}|&=\left|\tr  f\left(\frac{P_{\mu_{\min}}}{\tr(P_{\mu_{\min}} )} -  \frac{P_{\mu_{\max}}}{\tr(P_{\mu_{\max}} )}  \right)   \right|\\
&  =  \left|\tr \left( f\int_0^1  \dot{\gamma}(s)  ds\right)   \right| \\
&	=\left|  \int_0^1  \sum_{j\in\cJ} c_j \langle \partial_j f,[\gamma(s)]_{\omega_j}\partial_j U(s)\rangle ds    \right|\\
&\le\sqrt{( D^2+\eps)}\left( \int_0^1 \|\nabla f\|^2_{\LL,\gamma(s)}ds\right)^{1/2}\le \sqrt{(D^2+\eps) 2\lambda \e},
\end{align*}
where in the last line we used the Cauchy-Schwarz inequality with respect to the inner product $\sum_{j\in \cJ}c_j\langle .~,\int_0^1[\gamma(s)]_{\omega_j}ds~.\rangle $, and the result directly follows.
	\qed
\end{proof}

\section{From Ricci lower bound to modified log-Sobolev inequality}\label{mls}
In \cite{erbar2016poincar}, a modified logarithmic Sobolev inequality was proved to holds under the conditions that $\operatorname{Ric}(\LL)\ge 0$ and of boundedness of the diameter of the underlying space under the modified Wasserstein distance. Here, we extend their results to the quantum regime under the further assumption that the semigroup $(\Lambda_t)_{t\ge 0}$ is unital, leaving the study of the general case to later. The idea of the proof is to get a non-tight logarithmic Sobolev inequality from $\operatorname{HWI}(0)$, and then to tighten it using ideas borrowed from \cite{[BK08]}. In what follows, we denote by $d$ the dimension of $\cH$. \\\\
Given two states $\rho,\omega\in\cD(\cH)$, with associated spectral decompositions $\rho =\sum_{i\in \mathcal{A}}  \lambda_i P_i$, $\omega=\sum_{j\in \mathcal{B}} \mu_jQ_j$, where $\cA$ and $\cB$ are two finite index sets, a \textit{coupling} of $\rho$ and $\omega$ is a probability distribution $q$ on $\cA\times \cB$ such that 
\begin{align*}
	\sum_{i\in \cA} q(i,j) =\mu_j\tr(Q_j)\\
	\sum_{j\in\cB}q(i,j)=\lambda_i \tr(P_i).
\end{align*}
The set of couplings between $\rho$ and $\omega$ is denoted by $\Pi(\rho,\omega)$. In analogy with the classical literature (see e.g. \cite{[EM12]}), given an ergodic semigroup $(\Lambda_t)_{t\ge 0}$ with associated generator $\LL$, the \textit{coupling Wasserstein distance of order two} between $\rho$ and $\omega$ is defined as follows:
\begin{align*}
	W_{2,\LL,c}(\rho,\omega)^2:=\inf_{q\in \Pi(\rho,\sigma)}  \sum_{i\in\cA,~j\in\cB} q(i,j) W_{2,\LL} (\rho_j,\omega_j)^2,
\end{align*}	
where
\begin{align*}
	\rho_i:= \frac{P_i}{\tr P_i},~~~~~~~~~~\omega_j:=\frac{Q_j}{\tr(Q_j)},~~~~~~~~~i\in\cA,~j\in\cB.
\end{align*}
\noindent
The following result is a quantum generalization of Proposition 2.14 of \cite{[EM12]}:
\begin{proposition}\label{prop1}
Let $(\Lambda_t)_{t\ge 0}$ be a primitive QMS, with unique invariant state $\sigma$ and associated generator $\LL$, satisfying the detailed balance condition. Then, for any $\rho,\omega\in\cD_+(\cH)$,
\begin{align*}
	W_{2,\LL}(\rho,\omega)\le 	W_{2,\LL,c}(\rho,\omega).
\end{align*}	
\end{proposition}
\begin{proof}
	Let $\rho=\sum_{i\in \mathcal{A}}  \lambda_i P_i$, $\omega=\sum_{j\in \mathcal{B}} \mu_jQ_j$ the spectral decompositions of the states $\rho$ and $\omega$. For $(i,j)\in\cA\times \cB$, define $\rho_i:= \frac{P_i}{\tr( P_i)}$, $\omega_j:=\frac{Q_j}{\tr(Q_j)}$, and let $\eps>0$. By definition of the Wasserstein distance $W_{2,\LL}$, there exists a curve $\gamma_{ij}:[0,1]\mapsto \cD(\cH)$ from $\rho_i$ to $\omega_j$ such that
\begin{align*}
	\int_0^1 \|\dot{\gamma}_{ij}(s)\|_{g_{\LL,\gamma_{ij}(s)}}^2 ds\le W_{2,\LL}(\rho_i,\omega_j)^2+\eps.
	\end{align*}
For any coupling $q:\cA\times \cB \to \RR_+$ of the states $\rho$ and $\omega$, define the path $(\gamma(s))_{s\in[0,1]}$ on $\cD(\cH)$ as
\begin{align*}
\gamma(s)=\sum_{i\in\cA,~j\in\cB}q(i,j) \,\gamma_{ij}(s).
\end{align*}	
	Therefore, $\gamma(0)=\rho$ and $\gamma(1)=\omega$. Now,
	\begin{align*}
		W_{2,\LL}(\rho,\omega)^2&\le \int_0^1 \|\dot{\gamma}(s)\|_{g_{\LL,\gamma(s)}}^2ds\\
		&\le \sum_{i\in\cA,~j\in\cB} q(i,j) \int_0^1  \|\dot{\gamma}_{ij}(s)\|_{g_{\LL,\gamma(s)}}^2ds\\
		&\le \sum_{i\in\cA,~j\in\cB} q(i,j) ~W_{2,\LL}(\rho_i,\omega_j)^2+\eps.
		\end{align*}
	where we used the convexity of $g_{\LL}$ in the second line (see equation (8.15) of \cite{Carlen20171810}). As $\eps$ was arbitrary, the result follows after optimizing over the couplings $q$.
	\qed
	\end{proof}
\noindent
In what follows, we restrict our analysis to the case of a primitive QMS $(\Lambda_t)_{t\ge 0}$ with unique invariant state $\mathbb{I}/d$ that satisfies the detailed balance condition. In order to prove the main result of this section, we need the follows two lemmas that are extensions of Lemmas 6.2 and 6.3 of \cite{erbar2016poincar}:
 \begin{lemma}\label{lemma1}
 Assume that $\operatorname{Ric}(\LL)\ge 0$ and $\operatorname{Diam}_{\LL}(\cD(\cH))\le D$. Then for any $\delta>0$ and $f\in\cB_{sa}(\cH)$ such that $\tr (f^2)=d$:
 \begin{align*}
 	D\left(f^2/d\| \mathbb{I}/d \right)\le \delta D^2 \operatorname{I}_{\mathbb{I}/d}(f^2/d)+\frac{1}{4d\,\delta }   \tr (f^2 \mathbf{1}_{[1,\infty)}(f^2))
 	\end{align*}
 \end{lemma}	
\begin{proof}
The case when $f$ is not of full support is trivial, as then $\operatorname{I}_\sigma(f^2/d)=\infty$. Without loss of generality, we assume that $f$ has full support, so that $f^2/d\in\cD_+(\cH)$. Write $f=\sum_{i\in \cA}\varphi(i) P_i$ the spectral decomposition of $f$, for some index set $\cA$. From $\operatorname{HWI}(0)$, and Young's inequality \reff{Young} with $c=\delta D^2$, $x=\sqrt{\operatorname{I}_\sigma(f^2/d)}$, and $y=W_{2,\LL}(f^2/d,\mathbb{I}/d)$: 
\begin{align*}
	D(f^2/d\|\mathbb{I}/d)\le \delta D^2 \operatorname{I}_{\mathbb{I}/d}(f^2/d)+\frac{1}{4\delta D^2} W_{2,\LL}(f^2/d,\mathbb{I}/d)^2.
\end{align*}	
	From \Cref{prop1}, for any coupling $q:\cA\times \cB\to \RR_+$ between $f^2/d$ and $\mathbb{I}/d$ such that $q(i,j)=0$ whenever $\varphi(i)^2\le 1$,
\begin{align*}
	D(f^2/d\|\mathbb{I}/d)& \le\delta D^2\operatorname{I}_{\mathbb{I}/d}(f^2/d) +\frac{1}{4\delta D^2}\sum_{i,j:~\varphi(i)^2>1} q(i,j) W_{2,\LL}\left(\frac{P_i}{\tr(P_i)},\frac{P_j}{\tr(P_j)}\right)^2, \\
	 &\le\delta D^2\operatorname{I}_{\mathbb{I}/d}(f^2/d) +\frac{1}{4d\,\delta} \tr(f^2 \mathbf{1}_{(1,\infty)}(f^2)),
\end{align*}
which is what needed to be proved.
	\qed
	\end{proof}
\begin{lemma}
	For any $A>1$ there exists $\gamma >0$ such that for any $f\in\cB_{sa}(\cH)$ with $\tr(f^2)=d$,
	\begin{align}
		&\frac{1}{d}\tr ( f^2 \mathbf{1}_{[A^2,\infty)}(f^2) )\le \left(\frac{A}{A-1}\right)^2 \operatorname{Var}_{\mathbb{I}/d}(f)\label{eq4},\\
		& D(f^2/d\|\mathbb{I}/d)\le \gamma \operatorname{Var}_{\mathbb{I}/d}(f)+\frac{1}{d}\tr(f^2 \log f^2 \mathbf{1}_{[A^2,\infty)}(f^2)).\label{eq5}
		\end{align}
\end{lemma}	
\begin{proof}
This is a direct rewriting of Lemma 2.5 of \cite{[BK08]}.
	\qed
	\end{proof}

\begin{theorem}\label{mlsi1}
	Let $(\Lambda_t)_{t\ge 0}$ be a primitive semigroup with unique invariant state $\mathbb{I}/d$ and associated generator $\LL$. Assume that $\operatorname{Ric}(\LL)\ge 0$ and that $\operatorname{Diam}_{\LL}(\cD(\cH))\le D$. Then $\operatorname{MLSI}(cD^{-2})$ holds, for some universal constant $c$.
\end{theorem}	

\begin{proof}
	Let $A>1$ and $f\in\cB_{sa}(\cH)$ of spectral decomposition $f=\sum_{i\in\cA}\varphi(i) P_i$, with $\tr( f^2 )=d$. Withot loss of generality, we can assume $f$ positive definite. Then, set $f_A:= f\vee A\equiv \sum_{i:~\varphi(i)  \ge A}\varphi(i) P_i + A \,\mathbf{1}_{(-\infty, A)}(f)$. Define the state $\rho_A=f_A^2/\tr(f_A^2)$. By \reff{eq5}, 
	\begin{align}\label{eq7}
		D(f^2/d\|\mathbb{I}/d)\le \gamma \operatorname{Var}_{\mathbb{I}/d}(f)+  \frac{1}{d}\tr (f^2 \log f^2 \mathbf{1}_{[ A^2,\infty  )}(f^2)).
	\end{align}
By \Cref{poinc0},
\begin{align}\label{eq8}
\gamma	\operatorname{Var}_{\mathbb{I}/d} (f)\le - 2\e D^2 \gamma  \frac{1}{d}\langle f,\LL (f)\rangle  \le \frac{\gamma \e D^2}{2}     \operatorname{I}_{\mathbb{I}/d}(f^2/d),
\end{align}	
where in the last inequality, we used the strong regularity of Dirichlet forms of unital semigroups (see \cite{[KT13]}). Moreover,
\begin{align}
	\frac{1}{d}\tr &(f^2\log f^2 \mathbf{1}_{[A^2,\infty)}(f^2)) = \frac{1}{d} \tr(f_A^2\log f_A^2)-\frac{1}{d} A^2\log A^2 \tr(\mathbf{1}_{(-\infty, A)}(f))\nonumber\\
	&~~~~~~~\le (1+A^2)D(\rho_A\| \mathbb{I}/d) +  \frac{ \tr (f_A^2)}{d}\log\left(\frac{\tr f_A^2}{d} \right) -\frac{A^2\log A^2}{d}  \tr(\mathbf{1}_{(-\infty, A)}(f)),\label{eq9}
	\end{align}
where in the last line we used that $\frac{1}{d}\tr(f_A^2)\le 1+A^2$. However, from \Cref{lemma1} applied to $f_A$, since $\operatorname{I}_{\mathbb{I}/d}(\rho_A)\le \frac{d}{\tr(f_A^2)} \operatorname{I}_{\mathbb{I}/d}(f^2/d)$ by convexity of monotone Riemannian metrics (see e.g. equation (8.16) of \cite{Carlen20171810}),
\begin{align}
	D(\rho_A \| \mathbb{I}/d)&\le \frac{d\,\delta\, D^2}{\tr(f_A^2)}\operatorname{I}_{\mathbb{I}/d}(f^2/d)+\frac{1}{4\tr(f_A^2 )\delta }\tr(f_A^2\mathbf{1}_{[\frac{1}{d}\tr(f_A^2)    ,   \infty)}(f_A^2))\nonumber\\
	&\le  \frac{\delta D^2}{A^2}\operatorname{I}_{\mathbb{I}/d}(f^2/d)+\frac{1}{4d\,\delta A^2}\tr(f^2\mathbf{1}_{[A^2    ,   \infty)}(f^2)),\label{eq10}
\end{align}	
	where in the last line we used that $A^2\le \frac{1}{d}\tr(f_A^2)$. Using \reff{eq4} and \reff{eq8} together with \Cref{poinc0},
	\begin{align*}
		D(\rho_A\|\mathbb{I}/d)&\le  \frac{\delta D^2}{A^2}\operatorname{I}_{\mathbb{I}/d}(f^2/d) +\frac{1}{4\delta(A-1)^2}\operatorname{Var}_{\mathbb{I}/d}(f)\\
		&\le \frac{\delta D^2}{A^2} \operatorname{I}_{\mathbb{I}/d}(f^2/d) -\frac{ \e D^2 }{2\delta(A-1)^2}  \frac{1}{d}\langle f,\LL (f)\rangle  \\
		&\le \left( \frac{\delta D^2}{A^2}  +\frac{ \e D^2 }{8\delta(A-1)^2}   \right) \operatorname{I}_{\mathbb{I}/d}(f^2/d)\\
	\end{align*}
Now, 
\begin{align*}
	&\frac{\tr(f_A^2)}{d}\log \left(	\frac{\tr(f_A^2)}{d}\right) -\frac{A^2\log (A^2)}{d}\tr(\mathbf{1}_{(-\infty,A)}(f))\\
	&\le\frac{1}{d}\left({A^2}\tr(\mathbf{1}_{(-\infty,A)}(f))  + \tr(   f^2 \mathbf{1}_{[A^2,\infty)}(f^2) ) \right)\log\left[ \frac{1}{d}\left(  {A^2}\tr(\mathbf{1}_{(-\infty,A)}(f))  +  \tr(   f^2 \mathbf{1}_{[A^2,\infty)}(f^2) ) \right)\right]\\
	&~~~~~~~~~~~~~~~~~~~~~~~~~~~~~~~~~~~~~~~~~~~~~~~~~~~~~~~~~~~~~~~~~~~~~~~~~~~~~~~~~~~~~~-\frac{A^2\log(A^2)}{d}\tr(\mathbf{1}_{(-\infty,A)}(f)) \\
	&= A^2 \frac{1}{d}\tr(\mathbf{1}_{(-\infty ,A)}(f))\log\left(  \frac{1}{d}   \tr(\mathbf{1}_{(-\infty ,A)}(f))   +\frac{\tr(f^2  \mathbf{1}_{[A^2,\infty)}(f^2)  )}{dA^2} \right)\\
	&~~~~~~~~~~~~~~~~~~~~~~~~~~~~~~+ \frac{1}{d}\tr(  f^2 \mathbf{1}_{[ A^2,\infty )}(f^2)  )\log\left(  A^2 \frac{1}{d}\tr(  \mathbf{1}_{(-\infty,A]}(f) )+\frac{1}{d}\tr(  f^2\mathbf{1}_{[ A^2,\infty )}(f^2) )   \right)\\
	&\le A^2 \log\left(  1+ \frac{\tr(   f^2\mathbf{1}_{[ A^2,\infty )} (f^2) )}{dA^2} \right)+\frac{1}{d}\tr(f^2\mathbf{1}_{[ A^2,\infty )}(f^2))  \log (1+A^2)\\
	&\le (1+\log(1+A^2))\frac{1}{d}\tr(f^2\mathbf{1}_{[ A^2,\infty )}(f^2) ),
\end{align*}
where in the fourth line we used that $\frac{1}{d}\tr(f^2)=1$. Using once more \reff{eq4} and \reff{eq8} together with \Cref{poinc0}, we find
\begin{align}\label{eq6}
		\frac{\tr(f_A^2)}{d}&\log \left(	\frac{\tr(f_A^2)}{d}\right) -\frac{A^2\log (A^2)}{d}\tr(\mathbf{1}_{(-\infty,A]}(f))\le  \frac{\e D^2 A^2(1+\log(1+A^2))}{2(1-A)^2} \operatorname{I}_{\mathbb{I}/d}(f^2/d)
		\end{align}
The result follows after combining \reff{eq6}, \reff{eq7}, \reff{eq8}, \reff{eq9} and \reff{eq10}.
\qed
\end{proof}

\section{Conclusion}

In this paper, we prove that a classical picture, relating various inequalities which are useful in the analysis of Markov semigroups, carries over to the quantum setting.
Classically, a key element of this picture is a geometric inequality called the Ricci lower bound. Functional and transportation cost inequalities, which play an 
important role in the study of mixing times of a primitive Markov semigroup and concentration properties of its invariant measure, can be obtained from this geometric inequality.
The connection between them is provided by an interpolating inequality called the HWI inequality. In this paper, we analyze a quantum version of the Ricci lower bound (due to Carlen and Maas \cite{Carlen20171810})
and show that it implies a quantum HWI inequality, from which quantum versions of the functional and transportation cost inequalities (which are relevant for the analysis of quantum Markov semigroups) follow. 

\paragraph{Acknowledgements}
The authors would like to thank Ivan Bardet for helpful discussions.
	
\bibliographystyle{abbrv}
\bibliography{library}
\end{document}